\documentclass[acmsmall,dvipsnames,screen]{acmart}
\setcopyright{rightsretained}
\acmJournal{PACMPL}
\acmYear{2025} 
\acmVolume{9} 
\acmNumber{OOPSLA2} 
\acmArticle{} 
\acmMonth{04}
\acmDOI{10.1145/XXXXXXX}

\usepackage{xspace}
\usepackage[frozencache,cachedir=.]{minted}
\usepackage{tikz}
\usepackage{hhline}
\usepackage{tcolorbox}
\usepackage{multirow}
\usepackage{wrapfig}
\usepackage{mathpartir}
\usepackage{cancel}
\usepackage{pgfplots}
\usepackage{adjustbox}
\usepackage[flushleft]{threeparttable}
\usepackage{bm}
\usetikzlibrary{arrows,arrows.meta,shapes,shapes.arrows,calc,positioning,math,fit,tikzmark,decorations.markings}
\makeatletter
\tikzset{
    block filldraw1/.style={
        draw, fill=light-gray-in-algo},
    block filldraw2/.style={
        draw, fill=ddarkgray},
        }
\makeatother
\usepackage{listings}
\usepackage{enumitem}
\usepackage{algorithm}
\usepackage[noend]{algpseudocode}
\usepackage{subcaption}
\usepackage{amsmath}
\usepackage{fontawesome}
\usepackage[utf8]{inputenc}
\usepackage{algorithm}

\makeatletter 
\def\arcr{\@arraycr}
\makeatother

\definecolor{shadecolor}{gray}{1.00}
\definecolor{ddarkgray}{gray}{0.4}
\definecolor{darkgray}{gray}{0.7}
\definecolor{light-gray-in-algo}{gray}{0.8}
\definecolor{light-gray}{gray}{0.9}
\definecolor[named]{ACMBlue}{cmyk}{1,0.1,0,0.1}
\definecolor[named]{ACMYellow}{cmyk}{0,0.16,1,0}
\definecolor[named]{ACMOrange}{cmyk}{0,0.42,1,0.01}
\definecolor[named]{ACMRed}{cmyk}{0,0.90,0.86,0}
\definecolor[named]{ACMLightBlue}{cmyk}{0.49,0.01,0,0}
\definecolor[named]{ACMGreen}{cmyk}{0.20,0,1,0.19}
\definecolor[named]{ACMPurple}{cmyk}{0.55,1,0,0.15}
\definecolor[named]{ACMDarkBlue}{cmyk}{1,0.58,0,0.21}
\definecolor[named]{PaleGreen}{RGB}{196, 255, 231}
\definecolor[named]{PaleOrange}{RGB}{255, 213, 169}
\definecolor{intnull}{RGB}{213,229,255}

\newcommand{\graybox}[1]{\colorbox{light-gray}{#1}}
\newcommand{\darkgraybox}[1]{\colorbox{darkgray}{#1}}

\newcommand{\etc}{\emph{etc}\xspace}
\newcommand{\ie}{\emph{i.e.}\xspace}

\newcommand{\eg}{\emph{e.g.}\xspace}
\newcommand{\Eg}{\emph{E.g.}\xspace}

\newcommand{\aka}{\textit{a.k.a.}\xspace}
\newcommand{\cf}{\textit{cf.}\xspace}
\newcommand{\wrt}{\emph{w.r.t.}\xspace}
\newcommand{\Iff}{\emph{iff}\xspace}

\newcommand{\tname}[1]{\textsf{#1}\xspace}

\newcommand{\popper}{\tname{Popper}}
\newcommand{\aspal}{\tname{ASPAL}}
\newcommand{\aspsyn}{\tname{ASPSynth}}
\newcommand{\shape}{\tname{ShaPE}}

\newcommand{\vcdryad}{\tname{VCDryad}}
\newcommand{\grass}{\tname{GRASShopper}}
\newcommand{\veri}{\tname{VeriFast}}
\newcommand{\precis}{\tname{Precis}}
\newcommand{\prolog}{\tname{Prolog}}
\newcommand{\clingo}{\tname{Clingo}}
\newcommand{\suslik}{\tname{SuSLik}}
\newcommand{\sling}{\tname{SLING}}

\newcommand{\spt}{\tname{SPT}}
\newcommand{\synbad}{\tname{Synbad}}

\newcommand{\tool}{\tname{Sippy}}
\newcommand{\ggen}{\tname{Grippy}}

\definecolor{pblue}{rgb}{0.13,0.13,1}
\definecolor{pgreen}{rgb}{0,0.5,0}
\definecolor{pred}{rgb}{0.9,0,0}
\definecolor{pgrey}{rgb}{0.46,0.45,0.48}

\definecolor{ckeyword}{HTML}{7F0055}
\definecolor{ccomment}{HTML}{3F7F5F}
\definecolor{cnumber}{HTML}{2A0099}

\lstdefinelanguage{SynLang}{
  keywords={new, let, if, else, null, return, while},
  ndkeywords={bool, int, void, loc, set, pred, where},
  mathescape=true,
  showspaces=false,
  escapechar=$,
  showtabs=false,
  breaklines=true,
  showstringspaces=false,
  breakatwhitespace=true,
  lineskip=-0.9pt,
  morecomment=[l]{//}, 
  morecomment=[s]{/*}{*/}, 
  basewidth={0.54em, 0.4em},%
  basicstyle=\footnotesize\ttfamily,
  keywordstyle={\color{ACMPurple}\ttfamily\bfseries},
  ndkeywordstyle={\color{pblue}\ttfamily\bfseries},
  commentstyle={\color{ccomment}\itshape},
  numbers=none,
  moredelim=**[is][\color{red}]{@}{@},
}

\lstdefinestyle{numbers}
{
  numbers=left,
  numberstyle=\scriptsize\sf,
  xleftmargin=15pt
}

\newcommand{\set}[1]{\left\{{#1}\right\}}

\let\origthelstnumber\thelstnumber
\makeatletter
\newcommand*\Suppressnumber{%
  \lst@AddToHook{OnNewLine}{%
    \let\thelstnumber\relax%
     \advance\c@lstnumber-\@ne\relax%
    }%
}
\newcommand*\Reactivatenumber{%
  \lst@AddToHook{OnNewLine}{%
   \let\thelstnumber\origthelstnumber%
   \advance\c@lstnumber\@ne\relax}%
}
\makeatother

\newtcbox{\tracebox}[1][]{on line,size=fbox,#1}

\newcommand{\mynext}{\ensuremath{\mathsf{next}}}
\newcommand{\myvalue}{\ensuremath{\mathsf{value}}}

\newcommand{\code}[1]{\lstinline[basicstyle=\small\ttfamily,mathescape=true,escapechar=$,]{#1}}
\newcommand{\codeinmath}[1]{\text{\small\ensuremath{\mathtt{#1}}}}

\newcommand{\pcode}[1]{\mintinline[fontsize=\small]{prolog}{#1}}

\newcommand{\sym}[1]{\langle \text{#1} \rangle}
\newcommand{\pre}[1]{\lceil \text{#1} \rfloor}

\newtheorem{theorem}{Theorem}[section]
\newtheorem{definition}{Definition}[section]

\setlist[itemize]{leftmargin=*}
\setlist[enumerate]{leftmargin=*}

\begin{document}

\title{Inductive Synthesis of Inductive Heap Predicates}

\author{Ziyi Yang}
\affiliation{%
  \institution{National University of Singapore}
    \country{Singapore}
}
\email{yangziyi@u.nus.edu}
\orcid{0000-0002-8015-7846}

\author{Ilya Sergey}
\affiliation{%
  \institution{National University of Singapore}
    \country{Singapore}
}
\email{ilya@nus.edu.sg}
\orcid{0000-0003-4250-5392}

\begin{abstract}
  We present an approach to automatically synthesise recursive
predicates in Separation Logic (SL) from concrete data structure
instances using Inductive Logic Programming (ILP) techniques.
The main challenges to make such synthesis effective are (1)~making it
work without negative examples that are required in ILP but are
difficult to construct for heap-based structures in an automated
fashion, and (2)~to be capable of summarising not just the
\emph{shape} of a heap (\eg, it is a \emph{linked} list), but also the
\emph{properties} of the data it stores (\eg, it is a \emph{sorted}
linked list).
We tackle these challenges with a new predicate learning algorithm.
The key contributions of our work are (a)~the formulation of ILP-based
learning only using positive examples and (b)~an algorithm that
synthesises property-rich SL predicates from concrete \emph{memory
graphs} based on the positive-only learning.

We show that our framework can efficiently and correctly synthesise SL
predicates for structures that were beyond the reach of the
state-of-the-art tools, including those featuring non-trivial payload
constraints (\eg,~binary search trees) and nested recursion
(\eg,~$n$-ary trees).
We further extend the usability of our approach by a memory graph
generator that produces positive heap examples from programs. Finally,
we show how our approach facilitates deductive verification and
synthesis of correct-by-construction code.

\end{abstract}



\maketitle


\section{Introduction}
\label{sec:intro}

Separation Logic (SL) is a popular Hoare-style formalism for
specifying and verifying imperative programs that manipulate mutable
pointer-based data
structures~\cite{reynolds2002separation,OHearn-al:CSL01}. 
SL has been successfully applied to a wide range of applications,
including program verification~\cite{Appel-al:BOOK14,Jacobs-al:NFM11}, static
analysis~\cite{Calcagno-Distefano:NFM11}, bug
detection~\cite{LeRVBDO22}, invariant inference~\cite{le2019sling,DBLP:phd/basesearch/Dohrau22},
program synthesis~\cite{WatanabeGPPS21,polikarpova2019structuring}, and
repair~\cite{Tonder-LeGoues:ICSE18,NguyenTSC21}.
The key to the practical success of SL is its ability to enable
\emph{compositional} reasoning about programs in the presence of
potential pointer aliasing by exploiting the \emph{locality} of common
heap-manipulating operations. 
%
%
To enable expressive specifications, Separation Logic offers a
powerful mechanism to declaratively describe the shape and data
properties of linked heap-based structures, such as lists and trees:
\emph{inductive heap predicates}.
%
%
%
Unsurprisingly, defining precise and useful inductive predicates for
non-trivial data structures in general requires a good grasp of the
structure's \emph{internal invariants}.
Most existing SL-based reasoning frameworks require
defining predicates \emph{manually};
a few come with a set of pre-defined predicates for the most commonly
used data structures~\cite{Calcagno-al:JACM11,le2019sling},---thus, limiting the
ability of those approaches to verify or generally utilise \emph{data-specific} program
properties, such as, \eg, correctness of searching
an element in a binary search tree.



The aim of this work is to offer a methodology for \emph{automatically
  synthesising} inductive predicates for linked structures, where not
only the \emph{shape} but also the \emph{properties} of the stored
data (\eg, a binary tree being \emph{balanced}) would be captured.
To achieve this goal, we develop an approach for inferring inductive
SL predicates by synthesising them from \emph{memory graphs}, \ie,
concrete examples of data structure memory layouts, as produced by
programs that generate them.
Our work is closely related to two research themes: (1)~synthesising
formal representations of data
structures~\cite{guo2007shape,zhu2016automatically,LPAR23:Learning_Data_Structure_Shapes,molina2021evospex},
and (2)~using machine learning to infer data structure
invariants~\cite{brockschmidt2017learning,molina2019training,DBLP:conf/spin/UsmanWWYDK19}.
Existing approachers either impose specific restrictions on the inputs
of the synthesiser by, \eg, requiring functions constructing the data
structure \cite{zhu2016automatically,molina2021evospex} or a large
number of both positive \emph{and} negative
examples~\cite{molina2019training,DBLP:conf/spin/UsmanWWYDK19}; or
produce weaker specification, \eg, only inferring the structure shape,
but not its
properties~\cite{guo2007shape,LPAR23:Learning_Data_Structure_Shapes}.

To deliver an effective solution to this problem, our key idea is to
consider inductive heap predicates as \emph{logic programs} in a
\prolog-style language, and concrete memory graphs as \emph{logic
  facts}.
This perspective allows us to cast predicate synthesis as a classic
instance of Inductive Logic Programming~(ILP)---synthesising a logic
program by generalising concrete examples and facts about concrete
data instances that are also defined as logic
programs~\cite{Muggleton91,CropperDEM22}.
That said, to harness the power of ILP for synthesising SL predicates,
we have to overcome the following two challenges:

\begin{enumerate}[label=\textbf{C\arabic*},topsep=2pt,leftmargin=17pt]
\item \label{c1} For \emph{effectively} learning logic programs, ILP
  requires both \emph{positive} and \emph{negative} examples; the
  representative (\ie, non-trivial) negative examples are essential to
  ILP (\cf~\autoref{sec:popper}) but difficult to acquire without a
  human in the loop (\cf \autoref{sec:generator} for a discussion).
\item \label{c2} Synthesis of predicates in Separation Logic with
  arbitrary data constraints features a large search space, making it
  difficult to be \emph{efficient}.
\end{enumerate}

\noindent
To address the challenge \ref{c1}, we propose a novel
\textit{positive-only learning} approach to infer the \emph{most
  specific} logic predicate from a set of positive examples by
incorporating as many of the available (yet non-redundant) pre-defined
constraints as possible.
This is achieved by (1)~eliminating logically \emph{redundant}
restrictions featured in generated predicate candidates (as, \eg, the
last one in the series \pcode{A < B, B < C, A < C}) and
(2)~by introducing the notion of \emph{specificity} that selects a
locally-optimal inductive SL predicate from a set
of candidates with no redundancies.

Having phrased the inductive heap predicate synthesis as a search
for a local optimum, we inevitably face its large
computational complexity, which brings us to the challenge~\ref{c2}:
\emph{efficiency} of the search.
We address this challenge by exploiting the nature of our target
domain, \ie, Separation Logic. The key insight that allows us to prune
many non-viable candidates is to perform early detection of
\emph{invalid} combinations of heap constraints, \eg, those implying
that the same symbolic heap location can be \code{null} and not
\code{null} at the same time.
Combined, our solutions to the challenges~\ref{c1} and~\ref{c2}
deliver an approach for effective and efficient inductive synthesis of
inductive heap predicates from concrete memory graphs.

In summary, this work makes the following contributions:
%

%
\begin{itemize}
\item The first inductive synthesis approach of inductive heap
  predicates with arbitrary data constraints, requiring only positive
  examples of memory graphs, achieved by (1)~positive-only learning
  for ILP, (2)~exploiting the domain-specific properties of SL.
\item \tool---an automated tool for synthesising SL predicates from
  memory graphs, showcasing the \emph{effectiveness} of our approach
  for synthesising non-trivial predicates in a series of benchmarks.
\item \ggen---a memory graph generator that can automatically produce
  positive examples for the predicate synthesis via \tool, given
  data-manipulating programs, which it uses as test oracles.
\item Demonstration of utility of \tool by (a)~learning the predicates
  for verification from real-world heap-manipulating programs by
  obtaining memory graphs automatically via \ggen, and (b)~using the
  synthesised predicates for automated deductive synthesis.
\end{itemize}

\section{Overview and Key Ideas}
\label{sec:overview}

In this section, we provide a brief outline of the basics of
Separation Logic (SL) and its inductive predicates. Next, we explain
how to use the existing ILP system \popper~\cite{cropper2021learning}
to synthesise such predicates from \emph{both positive and negative}
examples, and how we handle with learning from positive examples only.
We conclude with the high-level workflow of our SL predicate
synthesiser \tool.

\subsection{Inductive Predicates in Separation Logic}
\label{sec:sl}

\begin{figure}
  \centering
  \includegraphics[width=0.8\textwidth]{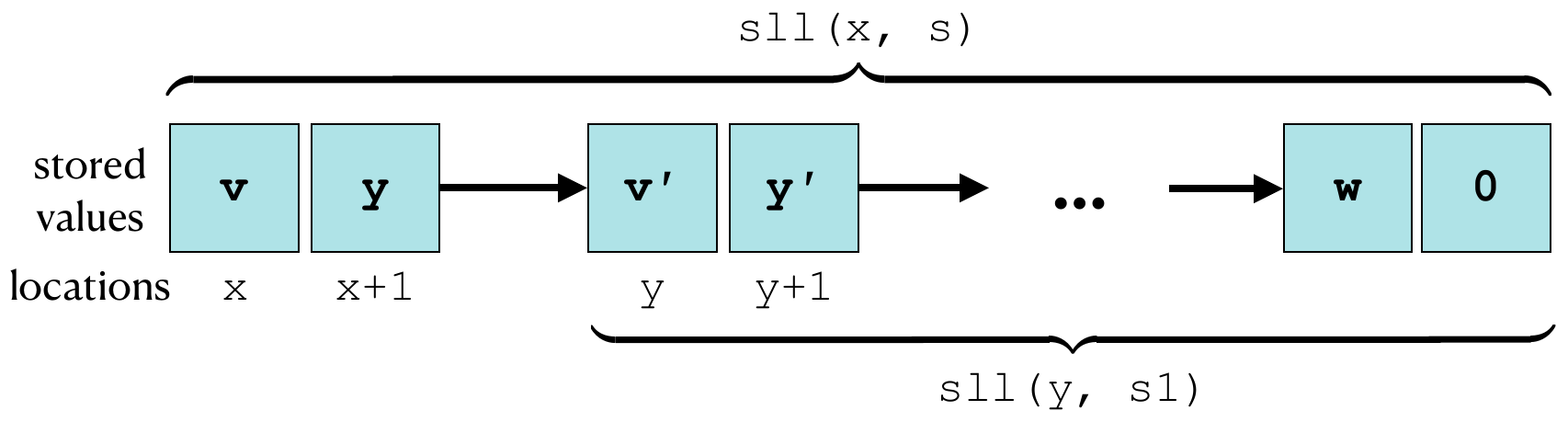}

  \caption{Memory layout of a singly linked list.}
  \label{fig:sll}
\end{figure}
Consider a schematic memory layout depicted in \autoref{fig:sll}
corresponding to a singly linked list (SLL).
The list has a recurring structure with each of its elements
represented by a consecutive pair of memory locations (the ``head''
one referred to by a pointer variable~\code{x}), the first one storing
its data value (or \emph{payload})~\code{v} and the second containing
the address \code{y} of the tail of the list. 
Knowing these shape constraints, the entire list can be traversed
recursively by starting from the head and following the tail pointers.

The idea of defining the repetitive shape of a heap-based linked
structure, such as SLL, is precisely captured by Separation Logic and
its inductive (\ie, well-founded recursive) predicates. One encoding
of an SLL heap shape via the SL predicate \code{sll} is given below:


\begin{lstlisting}
  pred sll(loc x, set s) where
    | x = 0 $\Rightarrow$ { s = {}; emp }
    | x $\neq$ 0 $\Rightarrow$ { s = {v} $\cup$ s1; x $\mapsto$ v * (x+1) $\mapsto$ y * sll(y, s1) }
\end{lstlisting}

\noindent
The predicate \code{sll} is parameterised by a location (\ie, pointer variable)
\code{x} and a payload set of the data structure
\code{s};
%
%
it holds true for any
\emph{heap fragment} that follows the shape of a linked list (and
contains no extra heap space).
What exactly that shape is, is defined by the two \emph{clauses} (\aka
constructors) of the predicate. 
The first one handles the case when \code{x} is a null-pointer,
constraining the payload set \code{s} of the list to be empty
(\code{\{\}}); the same holds for the list-carrying heap---which is
denoted by a standard SL assertion \code{emp}.\footnote{For
  simplicity, our examples use mathematical sets to encode the data
  payload, assuming uniqueness of the elements, instead of, \eg, an
  algebraic list. This is not a conceptual or practical limitation of
  our approach, as we will show in \autoref{sec:done}.}
The second clause describes a more interesting case, in which \code{x}
is not null, and so the payload can be split to an element \code{v}
and the residual payload \code{s1} (for simplicity of this example, we
assume that all elements of the list are unique).
Furthermore, the heap carrying the list is specified to have two
consecutive locations, \code{x} and \code{x + 1}, storing \code{v}
and some (existentially quantified) pointer value \code{y}, as denoted
by the SL \emph{points-to} notation $\mapsto$.
Finally, the rest of the SLL-carrying heap is 
the recursive occurrence of the same predicate \code{sll(y, s1)} (with
different arguments), thus replicating the recursive structure of the
layout from \autoref{fig:sll}.
%

The logical connective \code{*} appearing in the second clause of the
\code{sll} predicate is known as the \emph{separating conjunction}
(sometimes pronounced ``and separately'') and is the main enabling
feature of Separation Logic~\cite{OHearn-al:CSL01}.
It implicitly constrains the symbolic heaps it connects in a spatial
assertion to have \emph{disjoint} domains.
Specifically, in this example it implies that the heap fragment
captured by \code{sll(y, s1)} does not contain memory locations
referred to by either \code{x} or \code{x+1}.
Such disjointness constraint  is what makes it
possible to avoid extensive reasoning about aliasing when using SL
specifications, making them \emph{modular}, \ie, holding true in the
context of any heap that is larger than what is affected by the
specified program.
%



  
\subsection{From  Memory Graphs to Heap Predicates}
\label{sec:popper}

Our goal is to synthesise inductive SL predicates from examples of
concrete memory graphs.
To do so, we phrase both SL predicates and the memory graphs that
satisfy them in terms of Logic Programming. For example, the \prolog
predicate below defines a sorted singly linked list:
\begin{minted}[fontsize=\small]{prolog}
  srtl(X, S) :- empty(S), nullptr(X).
  srtl(X, S) :- next(X,Y), value(X,V), srtl(Y,SY), min_set(V,S), insert(SY,V,S).
\end{minted}
The predicate above defines a sorted singly linked list by enhancing
the ordinary singly linked list predicate with the constraint
\pcode{min_set(V, S)} that states that the value \pcode{V} is equal to
the smallest value in the set \pcode{S}. The \pcode{insert(S1, V, S)}
and \pcode{empty(S)} (\ie, \pcode{s == {v} ++ s1} and \pcode{s == {}}
in the SLL example) are defined using \prolog built-in predicates that
correspond to ordinary functions in set theory.
Other \prolog-style predicates used in the synthesised SL solutions
are data-structure specific and are extracted from the user-provided
memory graphs.
We leave till later (\autoref{sec:sldomain}) the issue of ensuring
that a \prolog predicate is also a \emph{valid} SL predicate in the
sense that it does not use SL connectives in a contradictory way,
allowing one to derive falsehood from its definition.

%
%

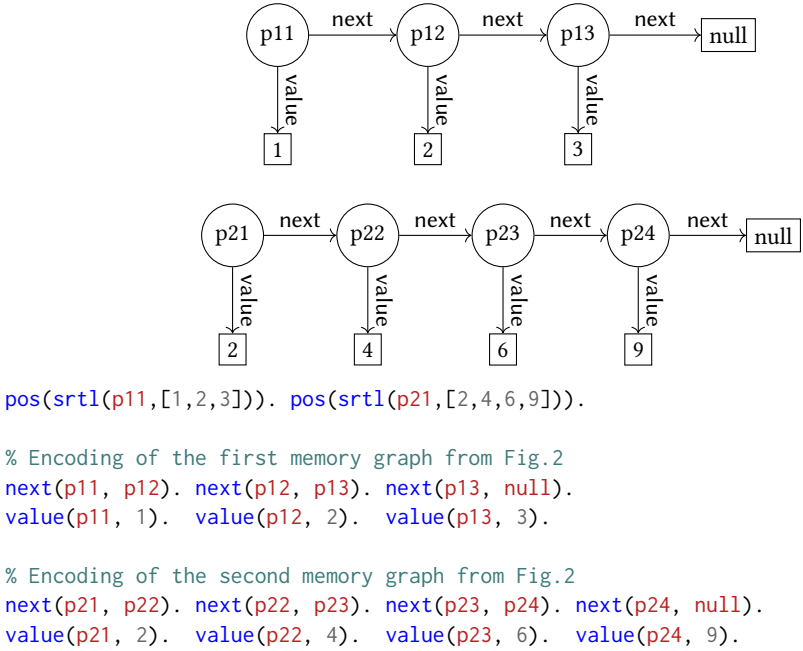
\begin{figure}[t]
\centering
{\small{
\begin{tabular}{c}
    \begin{tikzpicture}
        \node[circle, draw] (n1) at (0,1.5) {p11};
        \node[circle, draw] (n2) at (2,1.5) {p12};
        \node[circle, draw] (n3) at (4,1.5) {p13};
        \node[draw] (null) at (6,1.5) {null};
        \node[draw] (v1) at (0,0) {1};
        \node[draw] (v2) at (2,0) {2};
        \node[draw] (v3) at (4,0) {3};
        \draw[->] (n1) --  node [above,midway] {\mynext} (n2);
        \draw[->] (n2) --  node [above,midway] {\mynext} (n3);
        \draw[->] (n3) --  node [above,midway] {\mynext} (null);
        \draw[->] (n1) --  node [midway] [above,midway,sloped] {\myvalue} (v1);
        \draw[->] (n2) --  node [midway] [above,midway,sloped] {\myvalue} (v2);
        \draw[->] (n3) --  node [midway] [above,midway,sloped] {\myvalue} (v3);
    \end{tikzpicture}
\\
\\
  \begin{tikzpicture}
        \node[circle, draw] (n1) at (0,1.5) {p21};
        \node[circle, draw] (n2) at (1.8,1.5) {p22};
        \node[circle, draw] (n3) at (3.6,1.5) {p23};
        \node[circle, draw] (n4) at (5.4,1.5) {p24};
        \node[draw] (null) at (7.2,1.5) {null};
        \node[draw] (v1) at (0,0) {2};
        \node[draw] (v2) at (1.8,0) {4};
        \node[draw] (v3) at (3.6,0) {6};
        \node[draw] (v4) at (5.4,0) {9};
        \draw[->] (n1) --  node [above,midway] {\mynext} (n2);
        \draw[->] (n2) --  node [above,midway] {\mynext} (n3);
        \draw[->] (n3) --  node [above,midway] {\mynext} (n4);
        \draw[->] (n4) --  node [above,midway] {\mynext} (null);
        \draw[->] (n1) --  node [above,midway,sloped] {\myvalue} (v1);
        \draw[->] (n2) --  node [above,midway,sloped] {\myvalue} (v2);
        \draw[->] (n3) --  node [above,midway,sloped] {\myvalue} (v3);
        \draw[->] (n4) --  node [above,midway,sloped] {\myvalue} (v4);
    \end{tikzpicture}
\end{tabular}
}}
\\[5pt]
\begin{minted}[fontsize=\small]{prolog}
  pos(srtl(p11,[1,2,3])). pos(srtl(p21,[2,4,6,9])).

  % Encoding of the first memory graph from Fig.2
  next(p11, p12). next(p12, p13). next(p13, null). 
  value(p11, 1).  value(p12, 2).  value(p13, 3).

  % Encoding of the second memory graph from Fig.2
  next(p21, p22). next(p22, p23). next(p23, p24). next(p24, null).
  value(p21, 2).  value(p22, 4).  value(p23, 6).  value(p24, 9).
\end{minted}

    \caption{Positive examples of sorted list heap graphs, with the corresponding logic encoding.}
    \label{fig:sorted-list}
\end{figure}

The top of \autoref{fig:sorted-list} shows two memory graphs of sorted
lists that can be used to synthesise \pcode{srtl()}.
For a more natural representation in terms of Logic Programming, we
use Java-style naming of structure components, \ie, fields such as
{\small{\textsf{value}}} and {\small{\textsf{next}}} instead of
C-style integer pointer offsets;
these fields provide the data-specific predicates (\ie, \pcode{next()}
and \pcode{value()}) to the synthesiser. In the bottom of
\autoref{fig:sorted-list}, we provide the corresponding logic
representations of the inputs to the synthesiser, consisting of
positive examples (\ie, instances of the sought predicate that are
expected to be true), and the background knowledge (\ie, encoding of
the corresponding memory graphs) that should be used to derive those
examples using predicate candidates.
Given all this information, a synthesiser should be able to generate
the predicate \pcode{srtl()} that satisfies the positive examples.
That is, using the traditional program synthesis from
input-output pairs as an analogy~\cite{GulwaniPS17}, the facts in
background knowledge (\eg, \pcode{next(p11, p12)}) are inputs, the
examples (\eg, \pcode{pos(srtl(p11, [1,2,3]))}) are the outputs, and
the solution is the program (\ie, the predicate) to be synthesised.




\subsection{Predicate Synthesis via Answer Set Programming}
\label{sec:asp}

We observe that the synthesis of SL predicates can be regarded as a
Logic Programming synthesis task, studied extensively in the field of
Inductive Logic Programming
(ILP)~\cite{Muggleton91,cropper2020turning}. 
In ILP, the synthesis of definition is done by generating hypotheses
(\ie, predicates) and testing them against the provided examples.
Efficient generation of hypotheses in ILP is typically implemented
using Answer Set Programming (ASP)~\cite{gebser2022answer,aspguide}, a
constraint solving-based search-and-optimisation methodology that
allows for effectively pruning the search space of candidate
definitions and is used in many state-of-the-art ILP systems:
\aspal~\cite{corapi2011inductive}, \popper\cite{cropper2021learning},
and \aspsyn~\cite{bembenek2023smt}.

To see how ASP can be used for synthesising logic predicates, we first
provide a brief introduction to its principles using very basic
examples. 
Considered a declarative logic programming paradigm, ASP can be
regarded as a syntactic extension of \tname{Datalog}, but with a
different semantics called \emph{stable model
  semantics}~\cite{gelfond1988stable}. The output of an ASP program is
a set of \emph{models} (\ie, so-called answer set) that satisfy the
program constraints. A (normal) ASP program consists of a set of
\emph{clauses} that are composed of a head (on the left of
$\leftarrow$) and a body (on the right of $\leftarrow$) as:
\[
   a\ \leftarrow\ b_1,\ldots\ b_m,\ \neg\ c_1,\ \ldots,\ \neg\ c_n.
\]
which can be read as "if $b_1,\ldots\ b_m$ are true and
$c_1,\ \ldots,\ c_n$ are false, then $a$ is true". 
The statements $a,~b_i,~c_i$ are called \emph{literals} and are
declared in the format of \pcode{pred_name(X1, ..., Xn)} (\ie, a
predicate with arity $n$). 
A clause is called an \emph{integrity constraint} when its head (\ie,
the statement on the left-hand side of $\leftarrow$) is empty, which
means it is inconsistent if the body is true; a clause is called a
\emph{fact} when its body is empty, which means the head is always
true.

Instead of describing the formal definition of a stable
model, we show simple examples of ASP program and the corresponding
answer sets below:
\[
    \begin{tabular}{c|l|l}
      No. &ASP Program&Answer Sets\\\hline
      1&\pcode{a :- b. b.}&\pcode{{a,b}}\\
      2&\pcode{a :- not b. b.}&\pcode{{b}}\\
      3&\pcode{a :- not b. b :- not a.}&\pcode{{a},{b}}\\
      4&\pcode{a :- not b. b :- not a. :- a.}&\pcode{{b}}\\
    \end{tabular}
\]
The arity of the literals \pcode{a} and \pcode{b} is 0, and
\pcode{:-}, \pcode{not} in the programs mean $\leftarrow$ and $\neg$
correspondingly.
The programs in the table above and their answer sets should be
interpreted as follows.

\begin{itemize}
\item Program 1 is a simple program with a rule (general clause)
  \pcode{a :- b} and the fact \pcode{b} postulating that \pcode{b} is
  true. The answer set is \pcode{{a,b}} meaning that \pcode{a} and
  \pcode{b} can be true together, given the constraints.
\item Program 2 is similar to Program 1, but with the rule \pcode{a :-
    not b} instead, which means \pcode{a} is true when \pcode{b} is
  false. The answer set is \pcode{{b}}, no clause is making \pcode{a}
  true.
\item Program 3 is a program with two rules. The answer set is
  \pcode{{a},{b}} because \pcode{a} is true when \pcode{b} is false,
  and \pcode{b} is true when \pcode{a} is false, so the answer set is
  the combination of the two cases.
\item Program 4 is extended from Program 3 with another clause, that
  is an integrity constraint \pcode{:- a} forcing \pcode{a} to be false. The answer set is \pcode{{b}}
  because \pcode{b} is true when \pcode{a} is false, and the program is
  consistent only in this case (in contrast with Program~3).
\end{itemize}

\noindent
As Program 4 demonstrates, the integrity constraint can be used to
prune the answer sets---a very useful feature for synthesis tasks
(more discussion on ASP versus SMT is in \autoref{sec:related}).

Each program above can be regarded as an \emph{enumeration in the
  powerset} of the set with two elements, \pcode{a} and \pcode{b},
returning \emph{all sets} that satisfy the relations (\ie, the ASP
clauses) between the elements.
An ILP system, such as \popper~\cite{cropper2021learning}, relies on
an ASP solver to encode the enumerative search among all possible
combinations of literals to synthesise logic predicates.

As a concrete example of ILP via ASP, consider synthesising the
definition of a predicate \pcode{plus_two(A, B)} using six literals:
\pcode{succ(A, A)}, \pcode{succ(A, C)}, \pcode{succ(B, A)},
\pcode{succ(B, B)}, \pcode{succ(B, C)}, and \pcode{succ(C, B)} to
build the body of the predicate, with examples \pcode{plus_two(1,3)}
and \pcode{plus_two(2,4)}.
An ASP-based synthesiser would try to find a definition of
\pcode{plus_two(A, B)} as a suitable subset of all their $2^6=64$
possible combinations.
While doing so, it would also make use of the natural restrictions
that can be encoded as integrity constraints, such as
(1)~no free variable is allowed in the body (hence \pcode{{succ(A, B), succ(C, B)}} is
not a valid answer set because \pcode{C} is free), and 
(2)~all input variables \pcode{A} and \pcode{B} should appear in the
body (hence \pcode{{succ(B, C)}} is not a valid synthesis candidate).
As we will show, such constraints are also useful for encoding the
domain-specific knowledge about validity of SL predicates, and can be
efficiently solved by ASP solvers.
%

Moreover, the incrementality of ASP solvers make it possible to
constrain the search space continuously~\cite{gebser2019multi}. For
instance, assume the following hypothesis is obtained during the
search:
\begin{minted}[fontsize=\small]{prolog}
  plus_two(A, B) :- succ(A, A), succ(B, B).
\end{minted}
%

%
%
After testing it by \prolog against the examples, \popper finds that
none of the provided positive examples can be derived using this
solution.
As the result, other hypotheses that are more \emph{specialised} (\ie,
more constrained in the bodies) than it, such as the definition of
\pcode{plus_two()} below.
\begin{minted}[fontsize=\small]{prolog}
  plus_two(A, B) :- succ(A, A), succ(B, A), succ(B, B).
\end{minted}
will also entail no positive examples.
To this end, with the help of ASP, a classic ILP performs search for
a candidate hypothesis that passed all tests and has the smallest size
(\ie, number of literals in the predicate). Such \emph{optimal}
hypothesis for our example is the synthesised solution:
\begin{minted}[fontsize=\small]{prolog}
  plus_two(A, B) :- succ(A, C), succ(C, B).
\end{minted}

\subsection{Synthesis without Negative Examples}
\label{sec:approach}

The classic ILP comes with an important limitation: in general, it
requires both \emph{positive} and \emph{negative} examples to learn a
predicate. As we explain below, the need for the latter kind of
examples makes it challenging to employ ILP \emph{as-is} as a
pragmatic approach for synthesising SL predicates.

\paragraph{Why Negative Examples are Necessary in ILP}

Let us get back to our examples with synthesising a sorted singly
linked list predicate from positive examples of memory graphs in
\autoref{fig:sorted-list}. 
With the conventional ILP, the learned hypothesis by \popper is as
follows, and it is not what we need:
\begin{minted}[fontsize=\small]{prolog}
  srtl(X, S) :- empty(S), nullptr(X).
  srtl(X, S) :- next(X, Y), value(X, V), insert(SY, V, S), srtl(Y, SY).
\end{minted}
The learned hypothesis does not define a sorted list,
but an ordinary (unsorted) singly linked list.
The reason is: in the absence of the negative example, this is a
consistent hypothesis that is smaller in size than the correct
definition of \code{srtl}.
So if we want to learn the correct predicate, we need to provide negative examples that are inconsistent with the incorrect hypothesis, such as
\begin{minted}[fontsize=\small]{prolog}
  neg(srtl(p11, [1,2,3])).
  % Encoding of a negative example
  next(n1, n2). next(n2, n3). next(n3, null). 
  value(n1, 2). value(n2, 1). value(n3, 3).
\end{minted}
which is a singly linked but not sorted list. To summarise, when
performing synthesis via classic ILP, negative examples are necessary
to avoid the predicate being \emph{too general}.

\paragraph{Challenges in Obtaining Negative Examples.}
What makes things worse is that ILP systems rely on
\emph{representative} negative examples to correctly prune the
generality, which are hard to obtain automatically.
%
%
The difference between positive and negative examples is that, a good
set of positive examples (\(\mathbf{Pos}\)) only needs to guarantee
that all instances follow the predicate (\(\mathcal{P}\)), while a
good set of negative examples (\(\mathbf{Neg}\)) need to be much more
elaborated, so it could cover any possible way in which the predicate
can be wrong. This difference is expressed by the following
quantifications:
\[
  \forall \mathbf{e^+} \in \mathbf{Pos}, \mathcal{P}(\mathbf{e^+})
  \quad \text{vs.} \quad
  \forall \mathcal{P'} \subset \mathcal{P}, \exists \mathbf{e^-} \in \mathbf{Neg},  \mathcal{P'}(\mathbf{e^-}) \land \neg\mathcal{P}(\mathbf{e^-})
\]
That is, unlike a good positive example set, which is only quantified
over the examples, a good negative example set is quantified over the
\emph{predicates} and the examples, which makes it much harder to
achieve. As a concrete example of this phenomenon, imagine learning a
predicate for balanced binary trees.
A good set of negative examples would contain instances where
(a)~the height of the left subtree is too large, (b)~the height of the
right subtree is too large, (c)~the imbalance manifests recursively in
both left and right subtrees.
Without all these rather specific negative instances, it is possible
to learn a predicate with a constraint on the subtrees
\pcode{height(Left) <= height(Right) + 1}, which is not wrong but is
imprecise. This is not just an issue for SL predicates domain, but
also for general logical learning, witnesses by the fact that in
existing ILP benchmarks~\cite{cropper2021learning,thakkar2021example}
and specification mining framework \cite{10.1145/3622876} high-quality
negative examples are often crafted manually.

\begin{figure}[t]
  \centering
  \includegraphics[width=0.5\textwidth]{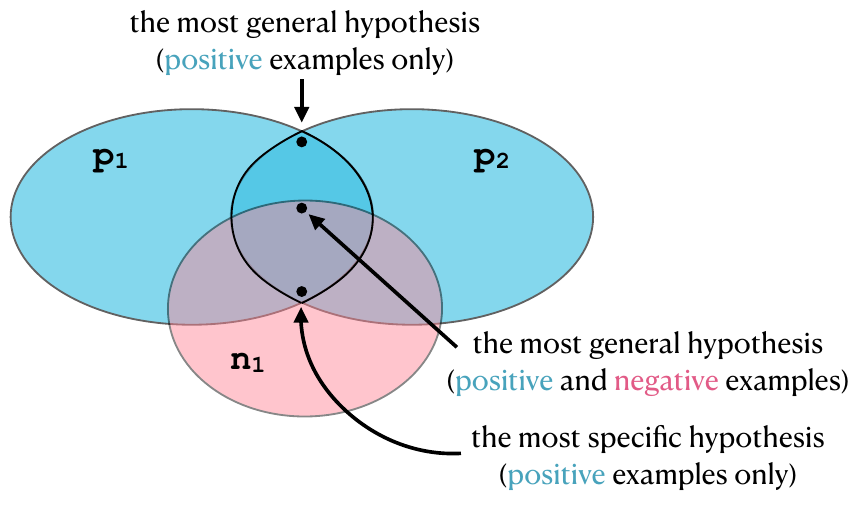}
  \caption{The effect of positive and negative examples on search.}
      \label{fig:illustration}
\end{figure}

This state of affairs brings us to the two key novel ideas of this
work that enable efficient synthesis of SL predicates only from
positive examples.

\subsubsection{Key Idea 1: Learning with Specificity.}
\label{sec:most-specific}

As discussed above, without negative examples a solution delivered by
\popper, while valid, may not be specific enough.
To provide more intuition on the space of possible design choices in
finding the best solution, together with the reason why positive-only
learning is possible, let us consider the illustration in
\autoref{fig:illustration}. The ``up''/``down'' in this figure
(informally) means ``more general/specific'', where the top/bottom
are constant True/False (\ie, the lattice is defined by
subsumption~\cite{muggleton1995inverse}).
Providing two positive examples, \code{p1} and \code{p2}, restricts
the search space for the solution hypothesis to the intersection of
their own spaces, with the most general one chosen as the solution.
Adding a negative example \code{n1} provides more restrictions, thus
allowing for more specific most-general solution.
From this diagram, one can see that, even without a negative example,
we can have a
\emph{tighter}~\cite{DBLP:journals/pacmpl/AstorgaSDWMX21} solution
(\ie, with stronger restrictions given the same number of clauses) if
we consider not the most general, but the most specific candidate in
the intersection of the search spaces defined by \code{p1} and
\code{p2} (generally, positive examples only).

Therefore, the basic idea of our positive-only learning is: to learn
\emph{the most specific} predicate admitting all provided examples.
The only problem is: what is the definition of ``specificity''? As the
opposite of ``generality'' (the program with the smallest number of
constraints), it is not practical to take the \emph{largest}
hypothesis as the most specific, as it would lead to \emph{redundant}
constraints.
As an example, consider the following valid formulation of the sorted
linked list predicate that requires, in its second clause, that
\code{T} = \code{SY}~$\cup \set{\codeinmath{V}}$ and \code{S} =
\code{T}~$\cup \set{\codeinmath{V}}$:
\begin{minted}[fontsize=\small]{prolog}
srtl(X, S) :- empty(S), nullptr(X).
srtl(X, S) :- next(X, Y), value(X, V), srtl(Y, SY), insert(SY, V, T), insert(T, V, S).
\end{minted}
Clearly, the last conjunction \pcode{insert(T, V, S)} is redundant and
can be removed because of the properties of the \pcode{insert(...)}
predicate.

To eliminate such candidates with redundancies, our approach for
positive-only learning encodes intrinsic logical properties of customised
predicates to \emph{minimise} the generated SL predicates.
Our tool comes with a pre-defined collection of properties of
predicates for common arithmetic (\eg, calculation, comparison) and
set operations (\eg, insertion, union) that are included into the
synthesis automatically.
More customised predicates can be added
by the user (with additional clause minimisation rules). 
After performing the minimisation hinted above (detailed in
\autoref{sec:normalise}), we define the \emph{specificity} of a predicate
candidate based on its size \wrt other available candidates
(\cf~\autoref{sec:tool}). The solution that is locally-optimal (\ie, the
strongest in the search space) will be adopted as the most specific predicate that is implied by
all the positive examples.

\subsubsection{Key Idea 2: Separation Logic-Based Pruning.}
\label{sec:pruning}

The domain of our synthesis, \ie, Separation Logic, provides effective
ways to prune the search space and accelerate the synthesis process. 
Postponing the detailed explanation of the optimisations until
\autoref{sec:SLsynthesis}, as an example, consider an important
property the separating conjunction stating that the fact
\code{x}~$\mapsto$~\code{a * y} $\mapsto$ \code{b} implies
\code{x}~$\neq$~\code{y} because of the disjointness assumption
enforced by \code{*}.
This property can be encoded as a pruning strategy via ASP integrity
constraints (\autoref{sec:asp}) that are generated by our tool for
each synthesis task.
%
%
Even for a \emph{doubly linked list}, one of the simplest predicates
in our benchmark (\cf~\autoref{sec:done}), without such optimisations,
the synthesis time is increased from 3 to 339 seconds; the synthesis
of more complex predicates does not terminate in 20 minutes without
SL-specific pruning.

\subsection{Automatically Generating Positive Examples}

So far, we assumed that the positive examples are provided by the
user, in the format of memory graphs. 
In practice, one may expect that such examples can be obtained in a
more automated way, \eg, from the available programs that manipulate
with the respective data structures.
For instance, an existing work on shape analysis~\cite{le2019sling} uses
a debugger for extracting memory graphs from a program's execution,
with the assumption that (1)~the user indicates the line of code to
extract the memory graph, and (2)~an input for the program is provided.
Unfortunately, this rules out a large set of programs that
\emph{expect} a data structure rather than generate one: without a
suitable input we simply cannot run them to obtain a graph, and constructing
such input is a task not much easier than encoding a memory graph manually.


To address this issue, we implemented \ggen---a tool that can
automatically generates positive example in the form of arbitrary valid
memory graphs of a data structure from the program that manipulates
with the structure, without requiring the user to provide any input.
The only assumption we make is that the program in question is
instrumented with test assertions, which can be used to filter out the
invalid memory graphs from the randomly generated ones. We further
show in \autoref{sec:verification} that \ggen can effectively generate
input graphs for synthesising SL predicates from real
heap-manipulating programs, reducing the specification burden for
proving their correctness.


\begin{figure}[!t]
  \centering
      
      \begin{adjustbox}{width=0.96\textwidth}
        \begin{tikzpicture}[
    node distance=2em and 3em,
    every node/.style={rectangle, draw, rounded corners, text centered, minimum height=1em},
    arrow/.style={-Stealth, thick},
    decision/.style={diamond, draw, text centered, inner sep=0pt, aspect=2},
]

\node[decision] (decision) {\faUser};
\node[above right=1em and 6em of decision] (generator) {\ggen ($\S\text{\ref{sec:generator}}$)};
\node[below right=1em and 6em of decision] (manual) {Write Manually};
\node[right=22em of decision] (sippy) {\tool ($\S\text{\ref{sec:positive} \&}~\S\text{\ref{sec:SLsynthesis}}$)} ;
\node[above right=1em and 12em of sippy] (verifiers) {Verification ($\S\text{\ref{sec:verification}}$)};
\node[right=12em of sippy] (synthesizers) {Synthesis ($\S\text{\ref{sec:synthesis}}$)};
\node[below right=1em and 12em of sippy] (others) {Other applications};

\draw[arrow] (decision) -- node[above, draw=none, align=center, sloped] {Program} (generator);
\draw[arrow] (decision) -- node[below, draw=none, align=center, sloped] {or} (manual);
\draw[arrow] (manual) -- node[above, draw=none, align=left] {Memory Graphs~~} (sippy);
\draw[arrow] (generator) --  (sippy);
\draw[arrow] (sippy) -- node[above, draw=none, align=center] {Heap Predicates} node[below,draw=none,align=center] {for further applications} ++(14em, 0) coordinate (branch);
\draw[arrow] (branch) |- (verifiers);
\draw[arrow] (branch) |- (synthesizers);
\draw[arrow] (branch) |- (others);

\end{tikzpicture}
      \end{adjustbox}
      \caption{The Workflow of \tool}
      
    \label{fig:extended}
\end{figure}
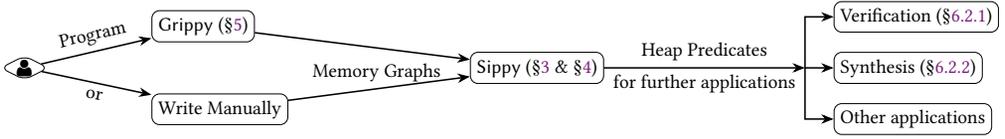


\subsection{Putting It All Together}

We implemented our algorithm for inductive synthesis of SL predicates
as the core of our tool called \tool. \autoref{fig:extended} shows the
high-level workflow of \tool. Starting from the left, the user can
provide positive examples (\ie, structure-specific heap graphs) for
the synthesiser by either using our graph generator \ggen on programs
that expect the data structure instance (\cf \autoref{sec:generator}),
or by manually writing them (\eg, in the style of Story-Board Programming
\cite{singh2012spt}).
Given the graph examples, \tool synthesises an inductive predicate
definition, which can be further used for program verification in SL-based
verifiers (\autoref{sec:verification}), for program synthesis
(\autoref{sec:synthesis}), or by any other SL-based tool.

\section{Positive-Only Predicate Learning}
\label{sec:positive}

In this section, we describe our approach for learning predicates from
positive-only examples not specific to SL predicates. To make the learning effective, we introduce the idea of predicate clause \emph{minimisation}, then explain how our specificity-based positive-only
learning works by showing its difference from the standard ILP systems \cite{cropper2021learning,cropper2022learning2}.

\vspace{-5pt}
\subsection{Normalising the Positive-Only Learning}
\label{sec:normalise}

As discussed in \autoref{sec:most-specific}, without negative
examples, adopting the smallest-size hypothesis, as done in many
program synthesisers~\cite{ji2021generalizable,lee2021combining}, can
lead to finding solutions that are too imprecise.
What about selecting the candidate with the largest size instead?
As shown by the {sorted list} (\pcode{srtl(X, S)}) example, more
literals in the body do not mean the predicate is more specific if the
body contains \emph{redundant} constraints among variables. Therefore, we first define a normalisation procedure to remove the redundancy in the predicate clauses with the help of the logical
\emph{entailment} (denoted as~$\models$ for \prolog clauses known as
\emph{definite clauses} in first-order logic).

\begin{definition}[Clause Minimisation]
  Let $\mathbf{numlit}(h)$ denotes the number of literals in the body
  of a clause $h$, and $ A \iff B$ denotes $A~\models~B$ and
  $B~\models~A$. A predicate clause $h$ is eliminated \Iff
  $\exists h_0,~h_0 = \arg\min_{X \in S} \mathbf{numlit}(X)~\land$
  $\mathbf{numlit}(h_0) < \mathbf{numlit}(h)$, where $S$ =
  $\{X \mid X \iff h\}$.
  
\end{definition}

In plain words, if for any clause $h$ there is a shorter clause $h_0$ equivalent (defined by entailments) to $h$, $h$ is
eliminated in the synthesis. Such minimisation is naturally encoded by ASP (detailed in \autoref{app:minimisation}).
Our experiments show that enumerable (also decidable) entailment rules
are effective for the synthesis of SL predicates within finite clause
length (\cf \autoref{sec:done}). Note that the minimisation is essentially a pruning for unnecessary clauses: we do not handle predicate-level redundancy (\aka Horn literal
minimisation~\cite{cepek1997stuctural}) because it is NP-hard~\cite{DBLP:journals/ai/HammerK93} in general.

\subsection{An Algorithm for Positive-Only Learning}
\label{sec:popper2}
\setlength{\textfloatsep}{0.5em} 
\begin{algorithm}[!t]
  \caption{The \darkgraybox{positive-only learning (POL)} algorithm v.s. the standard ASP-based \graybox{ILP}
    }
    \label{alg:popper}
    \begin{algorithmic}[1]
    \small
    \Require Search space: initial constraints \pcode{in_cons}, size limit \pcode{max_size}
    \Require User inputs: background knowledge \pcode{bk}, pos/neg examples \pcode{exs}
    \Procedure{\darkgraybox{POL} or \graybox{ILP}}{\pcode{bk, exs, in_cons, max_size}}
        \State \pcode{cons} = \pcode{in_cons}, \pcode{size} = 1, \pcode{sol} = \pcode{True}
        \While{\pcode{size}$ \neq$ \pcode{max_size}}
            \State \pcode{h} = \Call{Clingo}{\pcode{cons}, \pcode{size}} \Comment{\textit{generate}}
            \If{\pcode{h} is \textbf{UNSAT}}
                \State \pcode{size} += 1
            \Else
                \State \tikzmark{startp}\pcode{outcome} = \Call{Prolog}{\pcode{h}, \pcode{exs}, \pcode{bk}} \Comment{\textit{test}}
                \If{\pcode{outcome} is \textbf{complete} and \textbf{consistent}} 
                    \State \pcode{sol} = \pcode{h}
                    \State \textbf{break}
                \EndIf
                \State \tikzmark{startq}\pcode{cons} += \pcode{pruning(h, outcome)} \Comment{\textit{prune}\tikzmark{endp}}
                \If{\pcode{outcome} is \textbf{complete} and \pcode{sol} $\not\models$ \pcode{h}}
                    \State \pcode{sol} = \pcode{update_sol(h, comparable(h, sol))}
                \EndIf
                \State \pcode{cons} += \pcode{new_pruning(h, outcome)} \Comment{\textit{prune (updated)}\tikzmark{endq}}
            \EndIf
        \EndWhile
        \State \Return \pcode{sol}
    \EndProcedure
    \end{algorithmic}
    \begin{tikzpicture}[remember picture, overlay,opacity=0.4]
        \draw[block filldraw1] ($(pic cs:startp) +(-0.1, -0.1)$) rectangle ($(pic cs:endp) +(0, -0.1)$);
        \draw[block filldraw2] ($(pic cs:startq) +(-0.1, -0.1)$) rectangle ($(pic cs:endq) +(0, -0.1)$);
    \end{tikzpicture}
    \vspace{-1em}
\end{algorithm}

Our approach to learning predicates from positive-only examples can easily make use of components from the existing ASP-based ILP systems~\cite{cropper2021learning,cropper2022learning2}, seamlessly. In this section, we describe both the basic ILP learning loop by ASP and the procedure of our positive-only learning.

\autoref{alg:popper} (with the \graybox{light grey cover}) shows the
pseudocode of classic ILP learning loop, which follows the
``generate-test-constrain'' approach. It takes (1)~background
knowledge \pcode{bk}, (2)~positive and negative examples \pcode{exs},
and (3)~learning bias defining parameters of the search
\pcode{max_size}, such as the maximum size of a predicate to be
searched and other customisable parameters (\eg, whether to enable
mutually-recursive definitions), summarised as \pcode{in_cons}.
After encoding the iterative deepening search problem into an answer
set program, it first uses the ASP solver
\clingo~\cite{gebser2014clingo} to generate a hypothesis (line~4), and
then uses \prolog to test it against the provided examples (line~8).
A hypothesis is \emph{complete} if it entails all positive examples,
and \emph{consistent} if it entails no negative examples.
ILP returns a hypothesis as the solution if it is complete and
consistent (line~11); otherwise, it \emph{prunes} the search space
using the test outcome (line~12) and continues the search.
%
%
The test-based pruning works as follows.
Whenever hypothesis testing shows that not all positive examples are
true (incomplete), the pruning on \emph{specialisation} is applied;
whenever part of the negative examples are true (inconsistent), the pruning
on \emph{generalisation} is applied---ideas of the pruning are illustrated in
\autoref{sec:popper} and are detailed in~\cite{cropper2021learning}.
%


For the case of \emph{positive-only} learning
(\cf~\autoref{alg:popper} with the \darkgraybox{dark grey cover}), we
follow the approach from the work on \emph{specification
  synthesis}~\cite{park2023synthesizing} to output a set of
non-comparable (\ie, the element in the set is not entailed by any other) solutions as a conjunction: there are possibly several
solutions satisfying the positive examples where none of them is more
specific than the others, so the conjunction of them is the most
specific specification. 
The main difference from classic ILP starts from line~13, when a
hypothesis \pcode{h} is complete and is not entailed by the existing
solution set \pcode{sol}, the solution set is updated to adding
\pcode{h} to the set, together with removing all solutions in
\pcode{sol} that are comparable (semantically more general in our
case) to \pcode{h} (line~14).
The set of most specific solutions in the search space is returned
when the search is exhausted (line~16). The advanced pruning procedure
(line~15) is enabled by domain knowledge (\ie, Separation Logic
predicates in our case---\cf~\autoref{sec:sldomain}).

We conclude this section by proving the soundness of positive-only learning in \autoref{alg:popper}.

\begin{theorem}
  \label{thm:specific}
  The hypothesis set returned by the positive-only learning in
  \emph{\autoref{alg:popper}} contains all non-comparable, most specific predicates that is complete (\ie, entailing all examples) in the search
  space defined by the algorithm's initial constraints
  (\pcode{in_cons}) and the size limit (\pcode{max_size}) parameters.
\end{theorem}
\begin{proof}[Proof]
  By induction on the size limit \pcode{max_size} of the predicate: when \pcode{max_size} is 0, there is no predicate hypotheses, so \pcode{True} (the ``always true predicate'') is the only most specific one. Then assume that the theorem holds for \pcode{max_size} $n$, \ie, \pcode{sol_i} is the most specific hypothesis set; we prove it for \pcode{max_size} $n+1$.

  When \pcode{max_size} is $n+1$, based on the while loop in
  \autoref{alg:popper}, the search space for $n+1$ is the search space
  for $n$ plus when \pcode{size} is $n+1$. By the induction
  hypothesis, \pcode{sol_i} is the most specific hypothesis set in the search space
  for $n$, and the output \pcode{sol} is either \pcode{sol_i} or containing
  more specific predicates of space $n+1$ with all comparable predicates removed. Therefore, \pcode{sol} is the most
  specific hypothesis set in the search space with $n+1$ as \pcode{max_size}.

\end{proof}


\section{Separation Logic Predicate Synthesis via \tool}
\label{sec:SLsynthesis}

Having described the enhanced \emph{general-purpose} predicate
synthesis algorithm from positive-only examples,
we now show how to instantiate it for synthesis of inductive SL
predicates and improve the efficiency of the search algorithm by
exploiting domain-specific SL insights. We further discuss the
SL-validity of the synthesised predicates and the completeness of the
search algorithm.

\subsection{SL Predicates: Basics and Intricacies}
\label{sec:default}
 
\begin{figure}[!t]
  \centering
  \[
\begin{aligned}
  \sym{predicate} & ::= \sym{main\_pred} \;|\; \sym{main\_pred}  \sym{invented\_pred}\ast \\
  \sym{main\_pred} & ::= \sym{base\_clause}(\pre{main\_head}) \;|\; \sym{rec\_clause}(\pre{main\_head})\ast \\
  \sym{invented\_pred} & ::= \sym{base\_clause}(\pre{inv\_head}) \;|\; \sym{rec\_clause}(\pre{inv\_head})\ast \\
  \sym{base\_clause}(H) & ::= H(\codeinmath{This}, \sym{args}) \leftarrow \sym{base\_lit}\ast, \sym{pure\_lit}\ast \\
  \sym{rec\_clause}(H) & ::= H(\codeinmath{This}, \sym{args}) \leftarrow \sym{pointer\_lit}\ast, \sym{rec\_lit}\ast, \sym{pure\_lit}\ast \\
  \sym{literal}(R) & ::= R(\sym{args}) \\
  \sym{base\_lit} & ::= \sym{literal}(\pre{base\_pred}) \qquad\qquad \texttt{\% Pre-defined  for spatial relations} \\
  \sym{pure\_lit} & ::= \sym{literal}(\pre{pure\_pred}) \qquad \qquad\enspace \quad \texttt{\% Pre-defined  for pure relations} \\
  \sym{pointer\_lit} & ::= \pre{domain}(\codeinmath{This}, \sym{var}) \qquad\qquad\quad\enspace \texttt{\% Extract from the memory graphs} \\
  \sym{rec\_lit} & ::= \sym{literal}(\sym{head}) \\
  \sym{args} & ::= \sym{var} \;|\; \sym{var}, \sym{args} \\
  \sym{var} & ::= \codeinmath{X1} \;|\; \codeinmath{X2} \;|\; \dots \;|\; \codeinmath{This} \\
  \sym{head} & ::= \pre{main\_head} \;|\; \pre{inv\_head} \quad \texttt{\% From the task or randomly generated}
\end{aligned}
\]
\caption{The grammar of the SL predicates, in basic Backus–Naur form
  (BNF), extended with (1) meta-variables $(\cdot)$ for specialising
  the symbols, and (2) pre-defined atoms denoted by $\pre{X}$ (with
  comments of their origins).}
  \label{fig:grammar}
\end{figure}

We define the space of SL predicates in a way standard for
Syntax-Guided Synthesis (SyGuS)~\cite{Alur-al:FMCAD13}.
The grammar of the SL predicates is shown in \autoref{fig:grammar}. An
SL predicate is either having a shape with a single main predicate, or
shaped by a main predicate together with a set of invented
\emph{auxiliary} predicates, which are required in the case of nested
linked structures.

Specific to the predicates,
both main predicate and invented predicates consist of the base and recursive clauses, where the base clause is the one that does not have any recursive calls, and the recursive clause is the one that has recursive calls. The head literal (\ie, before $\leftarrow$) in each clause has a fixed argument \pcode{This} that denotes the base address of the data structure (similar to the \textit{this} reference in object-oriented programming).
The body literals (\ie, after $\leftarrow$) in the clauses are defined in terms of different predicates: the base (and pure) predicates are pre-defined, but extensible, to capture the spatial relation among the pointer for the base clause (the pure constraints among variables in clauses, respectively); the domain predicates describe the points-to relations can be obtained from the memory graphs; the recursive predicates are the recursive calls to the main or invented predicates.


Three aspects in the grammar in \autoref{fig:grammar} contribute to the
infinite synthesis search space: (1) the length of clauses, (2) the
number of sub-clauses for each predicate, (3) the arity of the
invented predicates. 
%
%
For our task, we noticed that predicates for real-world structures
rarely require more than 10 literals in their bodies; two sub-clauses
for each predicate are sufficient to capture the common structures;
and the arity of the invented predicates is set to be not more than
the arity of the main predicates. Such bounds for hypothesis space are
common in almost all synthesis-by-example tools~(\eg,
\cite{cropper2021learning,lee2021combining,Si-al:FSE18}), not only to
make the synthesis tractable, but also to avoid
overfitting~\cite{PadhiMN019} (\eg, a predicate disjointing facts of
all examples).

Below, we discuss two challenges in make SL predicate synthesis
effective and efficient, together with how we address them in \tool.



\subsubsection{Semantic-Based Pruning.}
\label{sec:semantics}

In most existing syntax-guided synthesisers \cite{cropper2021learning,Alur-al:FMCAD13,Si-al:FSE18}, the search is accelerated by pruning of the hypothesis search
space by employing the general \emph{syntax}
restrictions.
Other than limiting the syntax, we apply the following \emph{semantic}
properties' restriction of Separation Logic predicates to the search.
%
%
\begin{enumerate}
  \item \emph{Basic reachability}: no points-to relation appears in the
    body other than the ones from the \pcode{this} pointer. Thus, the clause \pcode{p(X, Y) :- next(X, Y), next(Y, Z), ...} is not 
    allowed as a candidate, because we expect all locations in the body to be
    accessible from \pcode{this} via fields.
  \item \emph{Basic assumptions}: the base (non-recursive) clause
    restricts \pcode{this} pointer to either be \code{null} or to equal to
    another pointer parameter variable. \Eg, \pcode{p(X, Y) :-
      nullptr(X), ...} is allowed, but \pcode{p(X, Y) :-
      next(X, Y), ...} cannot be a base clause.
  \item \emph{Restricted use of} \code{null}: if a variable \pcode{X} is
    a null-pointer (denoted by \pcode{nullptr(X)}), no
    more \pcode{X} occurs in the clause. \Eg, the clause \pcode{p(X, Y) :- nullptr(X), next(X, Y)}
    is not allowed.
  \item \emph{Quasi-well-founded recursion of payload}: the pure argument passed to a
    recursive call should (non-strictly) decrease. \Eg, for a clause
    \pcode{p(X, S) :- next(X, Z), p(Z, S1), ...}, the set \pcode{S} should contains \pcode{S1}. This
    is a common assumption in recursive program synthesis \cite{albarghouthi2013recursive,lee2021combining}.
  \item \emph{Heap functionality}: points-to relations of the same field
    should not target different locations. \Eg, a candidate clause cannot be \pcode{p(X, Y) :- next(X, Z), next(X, Z1), ...}.
  \end{enumerate}

\noindent
This list of search constraints represents a combination of the
properties implied by SL semantics (in a Java-style field-based memory
model) as well as by common properties of data structures, which are
essential for the efficient search of SL predicates.
The exact encodings of these constraints in ASP are provided and explained in \autoref{app:slsemantics}.

\subsubsection{Free Variables and Auxiliary Placeholders.}
\label{sec:auxiliary}

Free variables are common in SL predicates, \eg, the (implicitly
existentially-quantified) location \pcode{Y} in the base clause of the
 doubly linked list below:
\begin{minted}[fontsize=\small]{prolog}
  dll(X, Y) :- nullptr(X).
  dll(X, Y) :- next(X, Z), prev(X, Y), dll(Z, X).
\end{minted}
Unfortunately, completeness guarantees of pruning discussed in \autoref{sec:popper2}  do not hold for
predicates with free variables in the sense that
 a complete (\ie, valid) hypothesis with free
variables might  be wrongfully pruned during the search~\cite[\S{4.5}]{cropper2021learning}.
To address this problem, we introduce \emph{auxiliary placeholders}
into the search as a way to express predicate clauses with free
variables.
For example, the following doubly linked list predicate can be
regarded the same as the one above with \pcode{anypointer()}
placeholder, and \emph{can} be synthesised.
\begin{minted}[fontsize=\small]{prolog}
  dll(X, Y) :- nullptr(X), anypointer(Y).
  dll(X, Y) :- next(X, Z), prev(X, Y), dll(Z, X).
\end{minted}
On a technical level, this requires adding an ASP constraint (shown in \autoref{app:auxiliary})
that forces the parameter of the placeholder predicate (\pcode{Y}
here) to appear \emph{twice} in the whole clause, so it could be later
translated into a single occurrence of a free variable.

\subsection{Ensuring SL Validity in \prolog}
\label{sec:sldomain}

An astute reader can notice that the validity of the synthesised
predicates is not immediate due to our treatment of \prolog clauses as
SL assertions: the conjunction in \prolog does not guarantee the
\emph{separating conjunction} (\pcode{*}) in the SL sense. As an
example, consider the following simplified \prolog predicate for
binary trees:
\begin{minted}[fontsize=\small]{prolog}
  bi_tree(X) :- nullptr(X). 
  bi_tree(X) :- t1(X, L), t2(X, R), bi_tree(L), bi_tree(R).
\end{minted}
In this case, an instance of \pcode{bi_tree(X)} being evaluated to be
\emph{true} in \prolog can imply \emph{false} under SL semantics that
enforces heap disjointness: considering a memory graph with two nodes
\begin{minted}[fontsize=\small]{prolog}
  t1(n1, n2). t2(n1, n2). t1(n2, null). t2(n2, null).
\end{minted}
so that the graph fact \pcode{bi_tree(n1)} is provable in \prolog, but
the clauses \pcode{bi_tree(L)} and \pcode{bi_tree(R)} are
\emph{non-disjoint}.
Notice that, in our inductive synthesis setting, this situation would
correspond to having \emph{multiple} occurrences of the same points-to
fact in a memory graph representing a positive example for the
predicate, but should not be allowed by the definition of separating
conjunction.

To avoid this source of unsoundness, we use a straightforward solution
that enforces such separating conjunction semantics in \prolog: a
valid SL predicate is a complete \prolog predicate where the positive
examples succeed using each points-to fact \emph{exactly} one time (a
semantic property of SL assertions known as \emph{linearity}).
For the complete \prolog but invalid SL predicates, we also use the
\textit{specialisation} rule in \autoref{sec:popper2} to prune them:
if a predicate violates the linearity, then a more constrained one
will also violate it; this contributes to the new pruning in
line~15 of \autoref{alg:popper}.

We establish the following property of our SL-specific predicate
synthesis stating that, for the predicates in \tool's search space in
\autoref{sec:default}, if a memory graph is provable in \prolog with
linearity, then the corresponding heap is valid under SL semantics.

\begin{theorem}[SL Validity]
\label{thm:validity}
Let \pcode{F(h)} denote the memory graph of a heap \pcode{h}. For any
output predicate \pcode{p(X)} of \tool and any heap \pcode{h}, the
following fact holds: 
  \pcode{F(h)} $\models_{\prolog+\text{Lin}}$
\pcode{p(X)} $\Rightarrow$ \pcode{h} $\models_{\text{SL}}$ \pcode{p(X)}.
\end{theorem}

\subsection{The \tool Algorithm}
\label{sec:tool}

The only remaining step before putting all the pieces together is to
select the desired predicate from the set of non-comparable solutions
of positive-only learning. 
Even though predicates from POL can be conjuncted in general, the
semantics of SL predicates following the definition in
\autoref{sec:default} is more restrictive and the conjunction of valid
SL predicates may result in an ill-formed or a constantly false one. 
We found in practice that after the semantics-based normalisation from
\autoref{sec:normalise}, the number of literals can serve as a
\emph{good enough} specificity metric among incomparable predicates,
since containing more literals is likely to contain more information
or constraints about the heap structure. 
Following this intuition, we define the synthesis algorithm with
MAX\_POL function, which obtains the largest predicate from POL as per
\autoref{alg:popper}.

\begin{algorithm}[!t]
  \caption{The \tool loop for inductive predicate synthesis}
  \label{alg:sippy}
  \begin{algorithmic}[1]
  \small
  \Require memory graphs consist of \pcode{graph_bk, exs}
  \Procedure{Sippy}{\pcode{graph_bk, exs}}
      \State \pcode{graph_cons, shapes} = \pcode{graph_info(graph_bk, exs)}
      \State \pcode{max_var} = \pcode{max_body} = 1
      \State \pcode{sol} = \pcode{True} \Comment{The most general solution as initialisation}
      \For{\pcode{shape} in \pcode{shapes}}
        \State \pcode{max_size} = \pcode{maxsize(max_body, shape)}
        \State \pcode{h} = \Call{MAX\_POL}{\pcode{graph_bk, exs, graph_cons, max_size}}
        \If{\pcode{h} $\prec$ \pcode{sol}} \Comment{A more specific predicate is obtained}
            \State \pcode{max_var, max_body} = \pcode{(var_of(h), size_of(h))} + $\delta$
            \State \pcode{sol} = \pcode{h}
        \ElsIf{\pcode{sol} == True} \Comment{No predicate is yet learned}
            \If{\pcode{max_var} == \pcode{upper_bound}}
                \State \textbf{continue} \Comment{Try the next shape}
            \EndIf
            \State \pcode{max_var, max_body} += (1, 1)
        \Else
            \State \textbf{break} \Comment{No more specific predicate is found}
        \EndIf
      \EndFor
      \State \Return \pcode{sol}
  \EndProcedure
  \end{algorithmic}
\end{algorithm}

\autoref{alg:sippy} summarises the internal workings of \tool.
Our synthesiser takes as inputs memory graphs encoded as sets of logic
facts (\eg, \pcode{graph_bk}, such as \pcode{next(..)} and
\pcode{value(..)} from \autoref{fig:sorted-list}), positive examples of
heaps on which a predicate holds (\eg, \pcode{exs} as \pcode{srtl(..)}
from \autoref{fig:sorted-list}), so that the shape (matched with
pre-defined shapes in \autoref{sec:default}) a set of ASP constraints
(\pcode{graph_cons}) describing the information in the graphs (such as
the arity and types of the predicates to be learned) are obtained
(line~2).
Two parameters (line~3) for positive-only learning (MAX\_POL), (1)~the maximum number of
variables and (2)~the maximum size of the body of a predicate clause
for restricting the search space, are gradually increased during the
search using the following empirical strategy:
if no solution is valid (line~11), we either increase both parameters
by one to enlarge the space until finding one (line~14), or the
attempt on the current predicate shape fails (\ie, the upper bound of
the search space is reached), then
\tool will try synthesising using the next shape (line~13, \ie, more auxiliary predicates);
when obtaining one new better predicate than the existing, the search
parameters are both increased by a parameter~$\delta$ to find a
possibly more specific predicate (line~9), and the solution is
updated (line~10); if the learned predicate in the larger search space
is not better than the previous, we stop the search and output
(line~15-16).
The role of the parameter $\delta$ is, thus, to provide a ``margin''
for the completeness of the search: it is not guaranteed that \tool
will find the most specific solution \emph{across all possible search
  spaces}, but only in the search-space that is bound by the returned
output's parameters \emph{plus}~$\delta$.\footnote{We choose it to be (1,2) in our experiment from the natural observation: for our domain, we expect to have one body literal where the predicate is generating a new variable, and one more body literal that uses the new variable.}
Note that line~6 of \autoref{alg:sippy} features a function
\pcode{maxsize()} that calculates the maximum size of the search space based on the maximum number of variables and the predicate shape setting.

Finally, we provide a correctness argument for \tool. The soundness of
synthesising \emph{consistent} (\ie, inhabited) and \emph{well-formed}
(\ie, finitely provable) SL predicates is guaranteed by the soundness
of classic ILP and \autoref{thm:validity}. The following ``local''
completeness states that, given the output of \tool, no more specific
output can be discovered, \emph{even in} the larger search space
obtained by increasing the search parameters \emph{once} by $\delta$
at the line~9 of \autoref{alg:sippy}.

\begin{theorem}[Local Completeness of \tool]
\label{thm:completeness}
If the output of \tool is a predicate with the maximum number of
variables $m$ and the maximum length of the body $n$, then there is no
predicate with the maximum length of the body $m'$ and the maximum
number of variables $n'$, where $(m',~n')-(m,~n) = \delta$, that is
more specific than the output predicate.
\end{theorem}
\begin{proof}[Proof]
  Directly by contradiction and Theorem 3.1. Assume that the output solution \pcode{sol} is with size $(m,~n)$, and it is not the most specific one in size $(m',~n') = (m,~n) + \delta$.

 Because \pcode{sol} is the output, the search space is set to be $(m',~n')$ after the loop it is obtained. With Theorem 3.1 and the assumption, there is a solution \pcode{sol}$'$ in $(m',~n')$ that is more specific than \pcode{sol}, which is a contradiction with the output \pcode{sol}. Thus, \pcode{sol} is the most specific one in $(m',~n')$.
\end{proof}

\vspace{-5pt}

\section{Automated Heap Graph Generation via \ggen}
\label{sec:generator}

As presented, \tool requires input graphs to be provided explicitly.
The final technical contribution of this work is \ggen---an auxiliary
tool that allows one to \emph{automatically generate valid structure
  graphs} via programs that take them as inputs and are most likely
already available to the user.

Valid memory graphs can be obtained by extracting concrete
heap states from the program execution, provided concrete inputs, as,
\eg, done by the \sling tool~\cite{le2019sling}.
This approach is, however, only applicable if the program in question
\emph{generates} a data structure instance, and is problematic if it
\emph{expects} it as an input.
Our idea for automated graph generating is inspired by works on
fuzz-testing and uses a structure-expecting program as a
\emph{validator} for candidate graphs.
Given a program with assertions, \ggen generates \emph{arbitrary}
input memory graphs, subject to basic validity constraints (\eg, any
node should be reachable from its root), and keeps the graphs that
pass the assertions together with example facts summarised from those
graphs by traversing them and accumulating their payload in a set, so
they can be used as inputs for \tool.
The details are given in \autoref{alg:generateGraphs}.

One can argue that the generator defined this way can also effectively
produce \emph{negative} examples, thus removing the need for our
positive-only learning.
In practice, however, the number of the negative examples can be large
due to the randomness of the generator that has no knowledge of the
data structure (\cf~\autoref{sec:verification}). This makes it
impractical to use them for ILP-based synthesis that cannot
effectively discriminate good (\ie, informative) negative examples
from arbitrary junk.

\begin{algorithm}[t]
  \caption{Generating random memory graphs and examples}
    \label{alg:generateGraphs}
    \begin{algorithmic}[1]
    \Require $p$: Program with assertions as tests
    \Require $n$: Number of graphs to be generated
    \State Initialize $\mathit{GraphsAndExamples} \gets \emptyset$
    \For{$\mathit{sz} \gets 1$ to $max\_size$}
        \State $\mathit{cnt}$ = 0
        \While{$\mathit{cnt} < \lceil n/\mathit{max\_size} \rceil$}
        \State $\mathit{node}_1, \dots, \mathit{node}_{\mathit{sz}} \gets \mathit{init\_node}()$
        \State $\mathit{node}_i.key \gets \text{random(int)}$ for $i \in [1, \mathit{sz}]$
        \State $\mathit{node}_i.pointer \gets \text{random(node)}$ for $i \in [1, \mathit{sz}]$
        \If {$p(\mathit{node}_1, \dots, \mathit{node}_{\mathit{sz}})$ passes the tests}
        \State $\text{cnt} \gets \text{cnt} + 1$
        \State $\mathit{example} \gets \text{summarise\_graph}(node_1, \dots, node_{\mathit{sz}})$
        \State $\mathit{GraphsAndExamples}.\text{add}(\text{graph}(\mathit{node}_1, \dots, \mathit{node}_{\mathit{sz}}), \mathit{example})$
        \EndIf
        \EndWhile
    \EndFor
    
    \State \Return $\mathit{GraphsAndExamples}$
    \end{algorithmic}
\end{algorithm}


\vspace{-5pt}

\section{Evaluation and Discussion}
\label{sec:evaluation}

We implemented \tool by extending \popper with the combined use of
Python ($\sim$200~lines for modifying \popper), \prolog (14 rules for
specificity of predicates, and 19 supported predicates for pure
relations in first-order theories), and ASP ($\sim$200~lines for SL domain knowledge, and $\sim$300~lines for \tool's search space, where 46 rules are used to encode the minimisation rules of the 19 pure relations).
 The prototype of \ggen is implemented by $\sim$50~lines of ASP (to generate the graphs) together with $\sim$100~lines of Python. 
All experiments were conducted on
an 8-core M1 MacBook Pro with 16GB RAM. The valid structure-specific
input memory graphs were either manually/LLM-written, or automatically
generated via \ggen as described in \autoref{sec:generator}.




\subsection{Benchmarking Predicate Synthesis}
\label{sec:done}

To assess \tool's efficacy as a tool for heap predicate synthesis, our
evaluation addresses the following research questions:
\begin{enumerate}[label=\textbf{RQ 1.\arabic*},topsep=2pt,leftmargin=40pt]
\item\label{rq11} \emph{How effective is \tool in synthesising heap predicates?}
\item\label{rq12} \emph{What factors affect the synthesis efficiency?}
\item\label{rq13} \emph{How scalable is \tool with respect to the input size?}
\end{enumerate}

\subsubsection*{\ref{rq11}: Effectiveness and Expressiveness} 
We have assembled a set of benchmark for common and
complex heap-based data structures. \autoref{tab:predicates}
summarises our (mostly)\footnote{We will elaborate on the
  partially-successful binomial heap instance~\#15 in
  \autoref{sec:fail}.} successful case studies: 19~Separation~Logic predicates synthesised by \tool within the imposed 20~min
 time limit, which shows that \tool can effectively
produce SL predicates for a variety of heap-based data structures.

In our evaluation, we restricted the predicates in the search space to
feature at \emph{most one pure} theory, since for most predicates, the pure
relations on variables do not influence each other's validity. For
example, the \emph{AVL tree} predicate should feature (1)~a relation
between the heights of a node's subtrees and (2)~ordering relations on
the node payload; any one of these can be encoded independently,
leading to separate SL predicate definitions: \pcode{balanced(A, B)}
for \emph{balanced tree} and \pcode{bst(A, B)} for \emph{binary search
  tree}.
One can then \emph{merge} those two SL predicates with identical
spatial components into a united \emph{AVL tree} \pcode{avl(A, B, C)}
as follows:
\begin{enumerate}[topsep=2pt]
\item Merge the parameter lists by appending the pure variables;
\item Merge pure relations  with the overlapped variable renaming;
\item Adapt the recursion in the spatial parts for the new parameter.
\end{enumerate}

\begin{table}[t]
    
  \caption{Statistics on synthesised predicates and the comparison
    with other tool/setting. Columns following the predicate names are
    split by (1) the properties of the predicates: whether a predicate
    is invented (\ie, nested data structure, \underline{PI}), used
    pure relations \underline{Pure}, the size of the output predicate
    \underline{Size} (a triple of arity, the maximal number of
    literals, the maximal number of variables in the predicate); (2)
    the stats of \tool's synthesis: the percentage of run time taken
    by testing the hypotheses \underline{Test\%}, the synthesis time
    of finding the expected hypothesis \underline{T$_0$}, the
    synthesis time of exhaustive search in \textbf{P}ositive-only
    learning \underline{T$_\text{P}$}, 
    (3) the time \underline{T$_\text{I}$} it took to synthesise a
    result when using \popper as the classic \textbf{I}LP to
    synthesise, the time \underline{T$_\text{Io}$} when using classic
    \textbf{I}LP together with our \textbf{o}ptimisation to
    synthesise, and the time \underline{T$_\text{Pd}$} of running
    \textbf{P}OL with SL-based optimisations \textbf{d}isabled
    optimisations; finally, (4) we report the runtimes
    \underline{T$_\text{S}$} of \shape on the same tasks (but without pure relations). All times
    are in seconds, with TO for time-out, ERR for crashes, WA for wrong answers, and NA for not 
    supported in principle.}
\label{tab:predicates}
    \centering
    \begin{adjustbox}{width=0.98\textwidth} 
      

        \begin{tabular}{cl@{\ }|ccc@{\ }|@{\ }ccc|ccc|c}
          \toprule
          \textbf{No.} &\textbf{Predicate}& \textbf{PI} & \textbf{Pure} & \textbf{Size} & \textbf{Test\%} & \textbf{T$_0$}  & \textbf{T$_\text{P}$} &\textbf{T$_\text{I}$} &\textbf{T$_\text{Io}$} &\textbf{T$_\text{Pd}$} & \textbf{T$_\text{S}$}  \\
          \midrule
          1 &\texttt{singly linked list (payload)}& no & set&       (2,8,5)  & 2\% & $<$1 & 1  & 4 & <1 & 3 & \multirow{2}{*}{{{<1}}}  \\
          2 &\texttt{singly linked list (length)}& no & int&        (2,9,5)  & 1\% & 1 & 9  & NA & NA & 4 &   \\
          3 &\texttt{singly linked list segment}& no & set&         (3,8,6)  & 5\% & $<$1 & 1  & TO & 1 & 31 & 1  \\
          4 &\texttt{doubly linked list}& no & set&                 (3,10,6) & 1\% & 2 & 3  & NA & NA & 339 & 1 \\
          5 &\texttt{doubly linked list segment}& no & set&         (5,10,8) & 1\% & 49 & 192 & TO & 10 &  TO & TO \\
          6 &\texttt{sorted singly linked list}& no & set&          (2,9,5)  & 3\% & $<$1 & 1  & NA & NA & 4 & NA \\
          7 &\texttt{sorted doubly linked list}& no & set&          (3,11,6) & 2\% & 1 & 3  & NA & NA & 231 & NA \\
          8 &\texttt{circular list}& no & set&                      (3,8,6)  & 23\% & $<$1 & 1  & 18 & <1 & 81 & <1  \\
          9 &\texttt{lasso list}& no & set&                         (3,8,6)  & 25\% & $<$1 & 1  & TO & 1 & TO & NA \\
          10 &\texttt{binary tree}&  no & set&                       (2,11,8) & 8\% & 2 &13 & TO & 2 & 245 & WA \\
          11 &\texttt{back-linked tree}&  no & set&                  (3,13,9) & 9\% & 33 & 176 & TO & 41 & TO & TO   \\
          12 &\texttt{binary search tree (set)}& no & set&                 (2,13,8) & 15\% & 11 &26 & NA & NA & TO & \multirow{2}{*}{{{NA}}} \\
          13 &\texttt{binary search tree (list)}& no & list&                (2,13,8) & 39\% & 33 &433  & NA & NA & TO  \\
          14 &\texttt{balanced tree}& no & int&                        (2,12,7) & 32\% & 34 & 458  & NA & NA & TO & NA \\
          15 &\texttt{binomial heap (order)}& no  & int&             (3,14,8) & 33\% & 37 & 1,087  & NA &NA & TO & NA \\
          16 &\texttt{binomial heap (payload)}& no & set&            (3,14,9) & 2\% & 59 & 101  & NA & NA & TO & NA \\
          17 &\texttt{list of lists (set)}& yes & set&                     (2,17,5) & 1\% & 9 &321  & TO & 33 & TO & \multirow{2}{*}{{{ERR}}} \\
          18 &\texttt{list of lists (list)}& yes & list&                      (2,17,5) & 1\%  &14 &759 & TO & 47 & TO  \\
          19 &\texttt{rose (n-ary) tree}& yes & set&                 (2,17,6) & 1\% &9 &749 & NA & NA & TO & ERR \\
          \bottomrule
          \end{tabular}%

  \end{adjustbox}
    
%
\end{table}

To show that our positive-only learning (in \autoref{sec:positive})
and the SL-based optimisations (in \autoref{sec:semantics}) are
effective, we compared the synthesis time of \tool with the
unoptimised classic ILP system \popper, \popper with the SL-based
optimisations, and \tool without the optimisations (from left to
right in the table).
The results show that (1)~our positive-only learning effectively
learns the predicates which classic ILP cannot discover, and (2)~our
optimisations help to reduce the synthesis time significantly for both
classic ILP and POL for most cases. 
The only exception is \#2, where the unoptimised \tool is faster, is
due to the small size of the predicate, where the cost of adding
constraints for reducing the search space is larger than the benefit
of the pruning.
%

We attempted to compare synthesis capabilities of \tool to those
of \shape~\cite{LPAR23:Learning_Data_Structure_Shapes}, the most closely
related work.
Like \tool, \shape synthesises SL predicates for data structures from
memory graphs, allowing for positive-only synthesis without negative
examples via meta-interpretive learning
(MIL)~\cite{muggleton2014meta}.
However, unlike \tool, \shape only allows one to synthesise a structure's shape
constraints, without any restriction on the data payload or other pure constraints--this is why three data structures (SLL, BST, and a list of lists) with different pure relations are joined in \autoref{tab:predicates} into one predicate per structure in the \shape column \textbf{T$_\text{S}$}.
It has been shown in the prior work that the MIL-based learning cannot
define a complete search space for general logic
programs~\cite{cropper2015logical}, which indicates that \shape
\emph{cannot be extended}, even in principle, to express arbitrary
data constraints, such as arithmetic ones.
That is, theoretically, 9 (those not marked as NA) out of 16 predicates (with the joining taken into the account) in \autoref{tab:predicates} can be synthesised by \shape. However, in
practice, the bugs in \shape's learning loop implementation resulted in either (1)~the search timing out (TO), (2)~\shape terminating with an error (ERR), or (3)~producing overfitted predicates, \eg, 4 clauses instead of 2 for BST (WA).
At the end, we were able to only synthesise 4 predicates with \shape.

%


\subsubsection*{\ref{rq12}: Efficiency}
To understand what affects the synthesis efficiency, let us first look
into the case studies with long synthesis times (more than 5 minutes):
\#13 (BST with list payload), \#14 (balanced tree), \#15 (binomial
heap with order), \#18 (list of lists with list payload), and \#19
(rose tree with set payload).
All these case studies feature large predicate sizes, together with
either nested data structures or complex pure theory (integers or
lists). As witnessed by the last two columns of
\autoref{tab:predicates}, the time \tool takes to obtain the expected
predicate is less than 1 minute, and most runtime is spent on the
exhaustive search to give the guarantee of local completeness
(\cf~\autoref{thm:completeness}).
It is natural for nested data structures to take long time, because
the increment of search space is applied to both the
synthesis of the auxiliary predicate and the whole predicate.

We also wondered about how the used pure theories affect the
synthesis efficiency.
To answer that question, we compare two pairs of predicates (\#12 v.
\#13, \#17 v. \#18) using the same input memory graphs but different
pure theories: sets v. lists. The former is more efficient: this is
because the list is more expressive than set (\eg, a set union with
itself is eliminated, but appending a list to itself is producing a
new list), so the search space is larger after the redundancy
elimination. The same for integer theories: different
combinations of integer operations lead to large search space.

\subsubsection*{\ref{rq13}: Input and Scalability}
Our experiments demonstrate that 2-3 example graphs with tens of nodes
in total (less than 20 node per case study on average in our benchmark
suite) are sufficient to synthesise good predicates (we manually
assessed the results' quality).
We report the
fraction of time it took to test the hypotheses during the search.
%
%
In our case, it is not negligible for the examples with complex pure
relations (\eg, \#13-15).
In theory the test time should grow \emph{slower than linearly} with
the size of inputs. This is because only correct candidates require
traversing all nodes by the \prolog unification. Since most candidate
predicates do not satisfy all the examples, the testing is terminated
when \prolog reaches the node that falsifies the example.

\begin{figure}[!t]
  \centering  
  \begin{minipage}[t]{0.02\textwidth}
      \centering
      \begin{adjustbox}{width=0.8\textwidth}
        \rotatebox{90}{\small{\qquad\qquad\qquad Testing time (sec)}}
      \end{adjustbox}
  \end{minipage}
  \begin{minipage}[t]{0.45\textwidth}
      \centering
      \begin{adjustbox}{width=0.8\textwidth} 
      \begin{tikzpicture}
        \begin{axis}[
          xlabel={\large{Number of examples (100 nodes per graphs)}},
          grid=major,
          xmin=0, xmax=6,
          ymin=0, ymax=5,
          xtick={1,2,3,4,5},
          ytick={0,1,2,3,4,5},
          legend pos=north west,
          ]
          
          \addplot[line width=1pt, smooth, mark=*] coordinates {
            (1, 0.13)
            (2, 0.21)
            (3, 0.27)
            (4, 0.36)
            (5, 0.42)
          };
          
          \addplot[line width=1pt, smooth, mark=triangle, color=red] coordinates {
            (1, 0.18)
            (2, 0.33)
            (3, 0.50)
            (4, 0.46)
            (5, 0.55)
          };
    
          \addplot[line width=1pt, smooth, mark=diamond, color=green] coordinates {
            (1, 0.74)
            (2, 1.33)
            (3, 0.56)
            (4, 0.96)
            (5, 0.78)
          };
          
          \addplot[line width=1pt, smooth, mark=square, color=blue] coordinates {
            (1, 1.41)
            (2, 2.21)
            (3, 3.01)
            (4, 4.28)
            (5, 4.78)
          };
    
          \legend{SLL (len), Binary tree, Back-link tree, Lasso }
        \end{axis}
      \end{tikzpicture}
    \end{adjustbox}
  \end{minipage}
  \begin{minipage}[t]{0.44\textwidth}
    \centering
    \begin{adjustbox}{width=0.8\textwidth}
      \begin{tikzpicture}
        \begin{axis}[
          xlabel={\large{Nodes per graph (with 5 graphs)}},
          grid=major,
          xmin=0, xmax=120,
          ymin=0, ymax=5,
          xtick={20,40,60,80,100},
          ytick={0,1,2,3,4,5},
          legend pos=north west,
          ]
          
          \addplot[line width=1pt, smooth, mark=*] coordinates {
            (20, 0.08)
            (40, 0.13)
            (60, 0.21)
            (80, 0.31)
            (100, 0.42)
          };
          
          \addplot[line width=1pt, smooth, mark=triangle, color=red] coordinates {
            (20, 0.11)
            (40, 0.25)
            (60, 0.41)
            (80, 0.55)
            (100, 0.55)
          };
    
          \addplot[line width=1pt, smooth, mark=diamond, color=green] coordinates {
            (20, 1.01)
            (40, 1.67)
            (60, 0.65)
            (80, 1.17)
            (100, 0.78)
          };
          
          \addplot[line width=1pt, smooth, mark=square, color=blue] coordinates {
            (20, 0.31)
            (40, 0.63)
            (60, 1.46)
            (80, 2.69)
            (100, 4.78)
          };
    
          \legend{SLL (len), Binary tree, Back-link tree, Lasso }
        \end{axis}
      \end{tikzpicture}
    \end{adjustbox}    
\end{minipage}
\caption{Testing time  with different input graphs (left); and with different number
    of nodes per graph (right).}
    \label{fig:trend}
\end{figure}
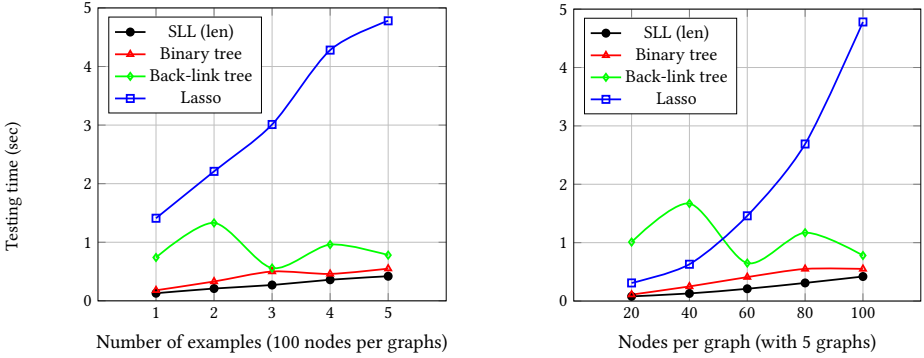

To show scalability of the tool \wrt the input size, 
\autoref{fig:trend} shows the testing time by graph number with
100~nodes per graph on its left, and the testing time by
node number with 5~graphs on its right. The general trends are as expected: the testing time
grows slower than linear in the graph number, and grows linearly in the node number (the percentage of nodes traversed in the examples should be constant with
the same topology of the graph, thus linear). 

Below we discuss the outliers. 
The reader can notice the substantial difference between
the testing time of the lasso list case study and every other example.
The reason for that is our current implementation of the linearity
check, which is not optimised for \emph{circular} data structures. 
Note that the testing of SL predicates consists of two parts: the
 \prolog validity check and the linearity check (in
\autoref{sec:sldomain}), we find that the all but circular data
structures' linearity checking time is negligible. This inefficiency
also explains the faster-than-linear growth of testing times in
\autoref{fig:trend}, and should be solvable by integrating the
checking into \prolog's SLD
resolution~\cite{DBLP:conf/ifip/Kowalski74} (left as
future work).
Finally, the testing time of back tree is not even strictly
increasing. The reason is: providing more or larger
examples makes it possible to prune
earlier, thereby reducing the overall synthesis time.
%

%


\subsection{The Utility of \tool}
\label{sec:utility}

\subsubsection{Verification}
\label{sec:verification}

Let us demonstrate how the combination of \tool and \ggen facilitates
deductive program verification in SL-based provers. 
Specifically, our goal is to streamline the task of writing inductive
predicates, by instead generating them automatically from programs to
be verified.
We are interested in the following research questions \wrt the
effectiveness of the approach:

\begin{enumerate}[label=\textbf{RQ 2.\arabic*},topsep=2pt,leftmargin=40pt]
\item\label{rq21} \emph{How much human effort is required to infer the
    predicates?}
  \item\label{rq22} \emph{How effective is \ggen for producing
      positive examples?}
  \item\label{rq23} \emph{Are the inferred predicates the same as the expected
      (human-written) ones, so they can be used directly for the
      verification?}
  \end{enumerate}

  \begin{table}[t]
    \centering
    \caption{Statistics on generating input memory graphs for different
      predicates: the lines of codes of the \code{assert}-annotated
      function or its test, the number of graph instances
      (capped by 100) generated in 10 minutes, the number of asserted
      conjuncts, 
      the ratio of the number of valid graphs to all randomly
      generated graphs, and whether the original program is verified
      with the synthesised predicates against its original
      specification. 
      For \veri and \grass, \textbf{Y} means verified; we did not manage
      to compile \vcdryad, so \textbf{=} means equivalent to the
      expected predicate, and \textbf{?} means not equivalent and
      couldn't be checked.
  %
    }
    \label{tab:gen}
  {\small{
      \begin{tabular}{l|c|ccccc|c}
        \toprule
        Predicate &Programs &  LOC$_\text{prog}$ & LOC$_\text{test}$ & Assert & Num &Ratio & Verified?\\
        \midrule
        \multirow{2}{*}{{{Singly Linked List}}} & concat\textsuperscript{2} &  10& - & 0& 100 &100\%  & \textbf{Y} \\
        &reverse\textsuperscript{2} &  10& - & 0& 100 &100\%  & \textbf{Y} \\
        \midrule
        \multirow{4}{*}{{{Sorted List}}} & find\textsuperscript{1} &  10& - & 1& 100 &14.4\% & \textbf{=} \\
        &insert\_iter\textsuperscript{1} &  28& - & 2 & 100 &13.3\% & \textbf{=} \\
        &copy\textsuperscript{2} &  26& - & 2 & 100 &13\% & \textbf{Y} \\
        &double\_all\textsuperscript{2} &  25& - & 2 & 100 &12.1\% & \textbf{Y} \\
        \midrule
        \multirow{4}{*}{{{Doubly Linked List}}} &append\textsuperscript{1} &  11& - & 2 &47 & 2.9\% &  \textbf{=} \\
        &dispose\textsuperscript{2} &  9& - & 2 &41 & 2.5\% & \textbf{Y} \\
        &reverse\textsuperscript{3} &  16& - & 2 &47 & 2.9\% & \textbf{Y} \\
        &\graybox{insert\_front\textsuperscript{1}} &  - &10&  2 &47 & 2.9\% & \textbf{=} \\
        \midrule
        \multirow{4}{*}{{{Binary Search Tree}}} & find\textsuperscript{1}  & 16& - & 3& 100 & 11.5\% & \textbf{=} \\
        &insert\textsuperscript{1} &  22& - & 6 &100 & 8.0\% & \textbf{=} \\
        &free\textsuperscript{3} &  11& - & 3 &100 & 7.5\% & \textbf{Y} \\
        &remove\textsuperscript{3} &  43& - & 3 &100 & 12.3\% & \textbf{Y} \\
        \midrule
        \multirow{2}{*}{{{Binomial Heap (order)}}} &\graybox{find\_min\textsuperscript{1}} &  \multirow{2}{*}{- }& \multirow{2}{*}{26} &\multirow{2}{*}{5} &\multirow{2}{*}{18} & \multirow{2}{*}{1.1\%} & \multirow{2}{*}{\textbf{?}} \\
        &\graybox{merge\textsuperscript{1}} &  & &  &  & &  \\
        \bottomrule
              \multicolumn{8}{l}{}
        \\[-8pt]
        \multicolumn{8}{l}{{\textsuperscript{1} From~\vcdryad~\cite{vcdryad}\quad \quad \quad
        \textsuperscript{2} From~\grass~\cite{Piskac-al:TACAS14}\quad \quad \quad
        \textsuperscript{3} From~\veri~\cite{Jacobs-al:NFM11}}}
      \end{tabular}
  }}
  \end{table}
\subsubsection*{\ref{rq21}: Required Human Effort}

For this experiment, we adopted the case studies from benchmark suites
of three different deductive
verifiers~\cite{vcdryad,Piskac-al:TACAS14,Jacobs-al:NFM11}, containing
heap-manipulating programs for different linked data structures, many
of which come with non-trivial data constraints.
Our aim is (1)~to quantify the human effort for annotating selected
program(s) with assertions as oracles for graph generation, and (2)~to
confirm that \ggen can produce good-quality graphs for \tool to
synthesise the expected predicates within a reasonable time limit (10
min). For simplicity, we only consider the graphs without cycles
unless being specified (doubly linked list in our case study).

\autoref{tab:gen} shows the statistics for generating up to 100 valid
graphs, with up to five nodes, within the time limit of 10 minutes, to
infer data structure predicates.
In all cases, between 0 and~5 simple assertions (\ie, one single
comparison for most cases, except for BST, one of whose assertions is
expressed by a comparison function) are enough for capturing the properties
(\eg, line 7 in \autoref{fig:srtl}).

\subsubsection*{\ref{rq22}: Effectiveness of the Graph Generator}

As shown in the table, in each experiment \ggen produced at least 18
valid graphs, which is sufficient for \tool to infer the expected
predicates. Unsurprisingly, the throughput of the generator (\ie,
the Num column in \autoref{tab:gen}) correlates with the
complexity of the predicates: \eg, the binomial heap
instance needs to satisfy not only the order relation between the
nodes but also the heap property, which results in low chances
for the generator to produce valid instances; generally, the throughput is similar (ranged by randomness)
 across programs manipulating the same structure.
As expected, the fraction of valid graphs \wrt all randomly generated
can be very low for complex predicates, which shows the importance of
positive-only learning, as large numbers of trivially invalid graphs
would slow down the synthesiser.

%
Even though in principle \ggen could be used for any programs,
we found that writing assertions for certain functions
(\graybox{greyboxed} in \autoref{tab:gen}) is not an effective way to
generate the valid instances. 
This is because the annotated node might simply not traverse ``enough''
of the structure to perform the validation.
As an example, consider inserting a node to the
front of a doubly linked list:
\begin{minted}[fontsize=\footnotesize,linenos,numbersep=-10pt]{c}
    DLNode * insert_front(DLNode * x, int k) {
      if (x == NULL) {
        DLNode * head = (DLNode *) malloc(sizeof(DLNode));
        head->key = k;
        head->next = NULL;
        head->prev = NULL;
        return head;
      } else {
        if(x->next != NULL) assert(x->next->prev == x);
        DLNode * head = (DLNode *) malloc(sizeof(DLNode));
        head->key = k;
        head->next = x;
        x->prev = head;
        return head; 
      } 
    }
\end{minted}
Testing the following incorrect example of a DLL heap graph, produced
by \ggen, will not violate the assertion at line~9 of the code above.
\begin{minted}[fontsize=\small]{prolog}
  next(n1,n2).   next(n2,n3). next(n3,null).
  prev(n1,null). prev(n2,n1). prev(n3,null).
\end{minted}
The reason is: the assertion at line~9 is only checked for node
\pcode{n1}, but the offending node \pcode{n2} is not checked because
the function does not traverse the whole list.
As a solution, 
in cases when no functions manipulating with the structure traverse
the whole structure graph (so their LOCs are not informative, hence
``-'' in \autoref{tab:gen}), one can instead write a standalone
\emph{traversal function} to be used as an oracle. Sizes of those
additional functions are shown as \text{LOC$_\text{test}$} in
\autoref{tab:gen}.
%

We conclude that \ggen is a useful front-end to \tool, but its
effectiveness depends on the ``thoroughness'' of the structure
traversal done by the function that is used as an oracle.

\begin{figure}[b]
  \begin{minted}[fontsize=\footnotesize,linenos,numbersep=-10pt]{c}
      int sorted_find(SNnode * l, int k){
        if (l == NULL) {
          return -1;
        } else if (l->key == k) {
          return 1;
        } else {
          if (l->next != NULL) assert(l->key <= l->next->key);
          int res = sorted_find(l->next, k);
          return res;
        } }
  \end{minted}
  
  \vspace{5pt}  
  
  \begin{minted}[fontsize=\small]{c}
  define pred sorted^(a): 
   ((a l= nil) & emp) | 
   ((a |-> loc next: c; int key: e) * sorted^(c) & (e lt-set keys^(c)))
  \end{minted}
  \caption{An example of the input program with assertions and
    inferred predicate for sorted list in \vcdryad.}
    \label{fig:srtl}
    \end{figure}

\subsubsection*{\ref{rq23}: Quality of Inferred Predicates}

With memory graphs obtained automatically, we synthesised the
respective predicates with \tool and translated them into the syntax
of the corresponding SL-based verifiers (automatically or manually,
based on how complex the concrete verifier's language is).
Next, we verified the original programs with the inferred predicates,
thus, demonstrating that the inferred predicates are equivalent to the
expected ones.
An inferred sorted list predicate of \vcdryad, corresponding to the
example from \autoref{sec:popper}, is shown in \autoref{fig:srtl}. The
only failing case is the binomial heap with order constraints because
of the limitation of pre-defined predicates in \tool (further
explained in \autoref{sec:fail}), where the synthesised predicates are
partially correct but not strong enough, and need to be refined by
manually adding the missing constraints.

To summarise, we found the combination of \ggen/\tool effective for
automatically producing SL predicates equivalent to human-written ones
from either modestly-annotated programs to be verified, or with a help
of a simple human-written traversal procedure for the data structure.

%

  \begin{figure}[!t]
    \centering  
    \begin{minipage}[b]{0.56\textwidth}
        \centering
        {\footnotesize{
        
    \begin{tabular}{c| c|c|c| c}
    \toprule
    No. & Category & Program & Code/Spec & Time    \\	  
    \midrule
    1 & \multirow{4}{*}{{{Deallocate}}} & sll & 5.5x & 0.2s\\
    2 & & bst & 8.0x & 0.2s\\
    3 & & dll\_seg & 2.4x & 0.2s\\
    4 & & {multilist} & 16.0x & 0.3s\\
    \midrule
    5 &  \multirow{3}{*}{{{Copy}}} & lseg & 2.0x & 0.8s\\
    6 & & bst & 3.5x & 3.3s\\
    7 & & balanced tree & 3.0x & 1.9s\\
    \midrule 
    8 & \multirow{2}{*}{{{Size}}} & sll\_len & 2.1x & 0.4s  \\
    9 & & balanced tree & 3.5x & 0.6s  \\
    \midrule 
    10 & \multirow{8}{*}{{{Transform}}} & sll $\rightarrow$ dlseg & 2.5x & 0.4s  \\
    11 & & {srt\_dll $\rightarrow$ sll} & 3.1x & 7.4s  \\
    12 & & {dll $\rightarrow$ bst} & 15.0x & 42.8s  \\
    13 & & {btree $\rightarrow$ bktree} & 13.6x & 11.8s  \\
    14 & & {multilist $\xrightarrow{}$ sll} & 5.0x & 8.8s  \\
    15 & & {btree $\xrightarrow{}$ dll} & 9.6x & 7.1s  \\
    16 & & {bst $\xrightarrow{}$ srtl} & 11.6x & 10.3s  \\
    17 & & {dll $\xrightarrow{}$ srt\_dll} & 7.3x & 9.3s  \\
      \bottomrule
    \end{tabular}
    

        }}
      \caption{Example programs synthesised by \suslik from SL
        specifications stated using predicates produced by \tool.}

      \label{tab:results}
    \end{minipage}
    \hfill
    \begin{minipage}[b]{0.4\textwidth}
      \centering
      
  \begin{minted}[fontsize=\scriptsize]{c}
  // pre:  {f :-> x ** sorted_dll(x, z, s)}
  // post: {f :-> y ** sll(y, s)}
  
  void srt_dll_to_sll (loc f) {
  loc x1 = READ_LOC(f, 0);
  if (x1 == 0)  {
    WRITE_INT(f, 0, 0);
    return;
  } else {
    int vx11 = READ_INT(x1, 0);
    loc nxtx11 = READ_LOC(x1, 1);
    loc z1 = READ_LOC(x1, 2);
    WRITE_LOC(f, 0, nxtx11);
    srt_dll_to_sll(f);
    loc y11 = READ_LOC(f, 0);
    loc y2 = (loc)malloc(2 * sizeof(loc));
    free(x1);
    WRITE_LOC(f, 0, y2);
    WRITE_LOC(y2, 1, y11);
    WRITE_INT(y2, 0, (int)vx11);
    return;
  }}
  \end{minted}
  \caption{An example \suslik output (\#{11}): a C program for
    converting a sorted DLL to SLL. \texttt{loc}, \texttt{READ\_LOC},
    \etc are macro-definitions around ordinary C types and operations. }
          \label{fig:transform}
  \end{minipage}
  \end{figure}

\subsubsection{Deductive Synthesis}
\label{sec:synthesis}

As another demonstration of \tool's utility, we employed the
synthesised SL predicates to automatically generate
correct-by-construction heap-manipulating programs in~C using a
state-of-the-art deductive
synthesiser~\suslik~\cite{polikarpova2019structuring,WatanabeGPPS21}.
The goal of this exercise was to demonstrate that, one can use \tool
together with \suslik to obtain \emph{provably correct}
implementations of structure-specific procedures for copying,
computing their size, and transformation without knowing how to
specify SL predicates, but using heap graphs. Our case study includes
17 synthesis tasks involving the predicates
from~\autoref{tab:predicates}, producing programs \emph{not} featured
in any past works on \suslik.
The average code/spec AST size ratio is 4.1, and the average \suslik
synthesis time is 6.2 sec, which shows that \tool produces predicates
that are \emph{immediately suitable} for proof-driven synthesis.
\autoref{tab:results} provides the detailed statistics. An example of
the synthesised C program that transforms a sorted DLL into a singly
linked list is given in \autoref{fig:transform}.

Compared to the existing example-based heap-manipulating program
synthesisers, \spt~\cite{singh2012spt} and
\synbad~\cite{DBLP:conf/sas/Roy13}, the joint \tool/\suslik workflow
does not require ``fold/unfold'' functions (as does \spt) or a
template (as needed by \synbad) for the intended programs.

      
  




\subsection{Failure Modes and Future Work}
\label{sec:fail}

In its current version, \tool failed to synthesise predicates for
several intricate linked structures. The reasons for the failed tasks fall into one of the following three categories:

\begin{enumerate}[leftmargin=*,topsep=2pt]
\item \tool's default settings cannot fully capture the structure's
  properties. Consider \#15 from \autoref{tab:predicates}: the root
  and leaf nodes of binomial trees have different order relation with
  their siblings, but our search space cannot express this
  distinction, so the synthesised predicate is \emph{the best in this
    setting} (capturing nodes' ordering \wrt their children) but not
  the expected one.

\item Nested data structures with multiple arguments or complex pure
  relations, For example, the needed search space of the \emph{braced list
    segment} predicate~\cite{Reynolds08} is too large to be fully
  explored within our time limits.

\item An instance of a predicate cannot be proven by top-\emph{down
  evaluation}, so that \prolog cannot evaluate them as expected in SL.
\end{enumerate}

\noindent
In the first case, the solution is naturally to allow richer search space: we might either extend the search space with more predicates (\eg, judging a node is root or not), or enabling larger parameters in \autoref{sec:default} to enrich the expressiveness of \tool.
This is in line with a common synthesis trade-off between the
expressiveness and the efficiency; we leave it a future work to find general
search space settings to enrich the expressiveness without much loss
of performance.

In the second case, it is possible to optimise the synthesis of nested
predicates using problem-specific knowledge. For instance, we can
assume the inner data structure is not mutually recursive (as it is in
the case of the braced list segment), reducing the search space by
splitting the synthesis of the auxiliary predicate and the whole
predicate. We leave this optimisation to the future work.

To explain the last issue, consider the following \prolog predicate:
\begin{minted}[fontsize=\small]{prolog}
  p(X, Y) :- X == Y.
  p(X, Y) :- next(X, Z), p(Z, Y), p(X, Z).
\end{minted}
In plain words: if \pcode{X} and \pcode{Y} have the same location,
then \pcode{p(X, Y)} is true; if not, then the \pcode{p(X, Y)} holds
if the both segments \pcode{p(X, Z)} and \pcode{p(Z, Y)} are true,
where \pcode{Z} is the next node of \pcode{X}.
This predicate is valid in Separation Logic, but \prolog rejects it,
because of its validity checking algorithm.
To understand the difference that causes the rejection, consider the
literal \pcode{p(a, b)}, in the case when \pcode{next(a, b)} is a fact
that holds.
In SL, the literal is true, because it is checked in a
\emph{bottom-up} way, \ie, ``whether the predicate is consistent if it is
true''. Therefore, \pcode{p(a, b)} holds because it is consistent with \pcode{next(a, b), p(b, b), p(a, b)}.
However, the test of \pcode{p(a, b)} in \prolog is done in a top-down way,
\ie, ``whether there is a variable unification that makes
\pcode{next(a, Z), p(Z, b), p(a, Z)} true'', which triggers a recursive
test on the inner  \pcode{p(a, b)}. It will be interesting to see whether replacing \prolog with the solvers from SL-COMP \cite{DBLP:conf/tacas/SighireanuPRGIR19} can directly solve it without other overhead.

\subsection{Why not just use Large Language Models?}

Though \tool has non-negligible runtime for synthesising complex
predicates, its completeness guarantees
(\cf~\autoref{thm:completeness}) ensure that the synthesised
predicates are the best (\ie, the most specific ones) in the
respective search space.
Large Language Models (LLMs), as a powerful tool for learning, have
been used extensively in the recent works for synthesising
specifications \cite{wen2024enchanting,ma2024specgen}, with faster
runtime but without completeness guarantees.
One may wonder: why not just ask an LLM to synthesise a specification,
and use \prolog to test it against the provided examples, mimicking the
loop of \autoref{alg:popper}?
To assess whether our synthesiser with proven completeness guarantees
provides better solutions compared to a state-of-the-art LLM, we pose
SL predicate synthesis as a task for the latter by designing a
detailed prompt outlining our intentions, followed by a series of
queries with inputs similar to what is required by \tool (\ie,
positive examples).
%


\subsubsection*{Phase 1: Simple Prompt for Learning}

Before the synthesis, we provide a detailed prompt to an LLMs as
outlinining the required background knowledge, which includes the
following parts:

\begin{enumerate}
\item The task: synthesising SL predicates in \prolog for linked heap
  structures given the graphs.
\item An example of the predicate ``sll'' for singly-linked list and
  its graphs, with the explanation. 
\item Other synthesis settings,
  such as the predicates
  that can be used for pure constraints, the option to invent
  auxiliary predicates, and the requirements on the size of the
  results. 
\end{enumerate}

\subsubsection*{Phase 2: Synthesising Complex Predicates}

A synthesis prompt for each task is given via the following template,
along with graphs and positive examples in the same format as taken by
\tool:

\begin{verbatim}
  For the next task, here are the graphs: (Graphs)
  And here are the positive examples: (Positive examples)
  Please, synthesise the predicate.
\end{verbatim}

\noindent
Our experiments were done on latest ChatGPT-4o.\footnote{The
  conversation snapshot is available
  at~\url{https://chatgpt.com/share/66f266a5-0338-8006-8bb9-1ef61c33d437}.}
%
%
Below, we summarise the outcomes.
\begin{enumerate}
\item We noticed that LLMs can correctly synthesise the predicates for
  simple cases (\eg, doubly-linked list), but it fails to synthesise
  predicates for more structures with non-trivial constraints (\eg, binary search
  trees and balanced tree): its result often don't type check or miss
  constraints.
\item Unsurprisingly, an LLM benefitted from the predicate names we
  provided to it. For example, with providing the name \code{rose_tree}
  in the prompt, the output predicates are mostly correct.
\item As LLMs have quick turnaround compared to running \tool, one
  promising direction is to use
  LLMs for the initial exploration of the search space and then use
  \tool to refine the results. 
\end{enumerate}

\noindent
We conclude that LLMs can be used for synthesising simple predicates,
and have a potential accelerate the synthesis process by deriving
plausible candidates, which can be checked by \prolog and, possibly,
repaired.
That said, completeness guarantees of \tool provide tangible benefits,
allowing it derive correct solutions for complex examples, which an
LLM failed to discover.

\section{Related Work}
\label{sec:related}


\paragraph{Learning Data Structure Invariants} 
%
%
%
%

Other than \shape~\cite{LPAR23:Learning_Data_Structure_Shapes},
discussed extensively in \autoref{sec:done}, earlier work on shape
analysis also used inductive synthesis to generate shape
predicates~\cite{guo2007shape}, but the input of the synthesis
framework is a program that constructs the data structure instance,
providing more information (\eg, the recursion structure) compared to
memory graphs.
Similar to \shape, that work only considers the shape relation without
the data properties. \tname{DOrder}~\cite{zhu2016automatically} and
\tname{Evospex}~\cite{molina2021evospex} are two later works on
learning the data invariants from the constructors of the data
structures (in OCaml or Java).
\tname{Locust}~\cite{brockschmidt2017learning} infers shape predicates
from pre-defined definitions with statistical machine learning, with
no completeness guarantees. The work by
\citeauthor{molina2019training}~\cite{molina2019training} describes a deep learning-based
framework that implements a binary classifier for the data structure
invariants; unlike our work, it does not
provide logical descriptions of data structures but merely tells valid
structure instances from invalid ones. Though more machine learning
methods \cite{DBLP:conf/spin/UsmanWWYDK19} have shown to be effective
in learning data structure, the training data is required to be large
and diverse, which is not always available in practise.
\citet{DBLP:phd/basesearch/Dohrau22} describes a black-box approach to
infer SL specifications and predicates from programs based on ICE
learning~\cite{garg2014ice}, while it can neither deal with, nor be
easily extensible to nested data structures. \tname{SLING} is a
framework to infer program specifications in Separation Logic from
memory graphs~\cite{le2019sling}; it does not infer new heap
predicates and instead offers a number of pre-defined structure
shapes, where \tool can work as a complementary tool to help the user
to pre-define new predicates.




\paragraph{Synthesising Declarative Representations}

The approach of \tool extends two lines of work on synthesis of
declarative representations programs and data.
The first one is inductive logic programming (ILP), which aims to
learn logic programs from examples.
\tname{Progol}~\cite{muggleton1997learning} is an early notable ILP
system that achieves positive-only learning by Bayesian framework,
which is not sound in general. Importantly, \tname{Progol} does not
support learning recursive logic programs.
\tname{AMIE}~\cite{galarraga2013amie} is another knowledge rule-mining
framework with positive-only examples, but the learning is in
\tname{AMIE} mainly targeted knowledge base graphs, so the learned
rules were not as complex as in ILP settings.
Other than those conventional ILP methods that synthesise \prolog
programs, significant progress has been made recently on synthesising
\tname{Datalog}~\cite{thakkar2021example,10.1145/3622847} and ASP
programs~\cite{DBLP:conf/jelia/LawRB14,DBLP:conf/aaai/LawRBB020} from
examples.
While the syntax of Separation Logic assertions and predicates can be
expressed in the \tname{Datalog} or ASP domain, the approaches
developed in the past efforts are not immediately applicable, as they:
(1)~may require negative examples (in the case of \tname{Datalog}
synthesisers), and (2)~limit the pre-defined predicates expressed only
by grounded facts.
In particular, the latter means that the predicates used in the
synthesis cannot express arbitrary operations, because grounded facts
are fully instantiated and do not contain variables. 
This limitation makes impossible the representation of general rules
or operations that can be applied to a range of inputs. For example, a
grounded fact can state that a specific element belongs to a set, but
it cannot express a general rule for membership that applies to
\emph{any} element. 
At the same time, in \tool, the predicates are written in \prolog, a
Turing-complete language, which enables definitions of operations like
list append, set union, \etc,---a feature our approach has directly
inherited from \popper~\cite{cropper2021learning}.

The second relevant line of work is \emph{specification synthesis},
which aims at synthesising formulas within various but pre-defined
domains. 
Our graph generator (\autoref{sec:generator}) follows the ideas of
\precis~\cite{DBLP:journals/pacmpl/AstorgaSDWMX21}, which similarly
uses test cases as a learning oracle to generate positive example
for synthesising program contracts. 
At the same time, the formal guarantees provided by positive only
learning--generating all non-comparable and most specific predicates,
are similar to those of \tname{Spyro}~\cite{park2023synthesizing},
which defines a general framework for synthesising specifications for
customisable domains. The main difference between \tool and those
works is that \tool synthesises \emph{predicates} that can contain
\emph{recursive definitions}---an aspect cannot be handled by the
existing specification synthesisers.

The overlap between two lines of work above is the notion of
\emph{least general generalisation}~(LGG)~\cite{plotkin1970note}. The
example-based specification synthesisers can be considered as learning
the LGG of the examples, which is exactly what early bottom-up ILP
systems do. The limitation of existing LGG operations is well-known in
ILP~\cite{CropperD22}: there is no LGG operations for recursive logic
programs, which is the reason why modern ILP systems are defined in a
top-down fashion.

\paragraph{Answer Set Programming v. Satisfiability Modulo Theories}

ASP plays a crucial role in our work, similar to that of SMT solvers
most contemporary synthesis tools. 
As mentioned by \citeauthor{bembenek2023smt}~\cite{bembenek2023smt},
ASP is effective at search tasks involving fixpoints: pruning in \tool
can be encoded easily with recursive logic predicates, some of which
though can be expressed in SMT, need to be encoded in a more complex
and hard-to-understand way.
Another advantage of ASP is its efficiency when enumerating \emph{all}
solutions in a search space, which is crucial for exhaustively exploring
the space of SL predicates. 
In contrast, in most state-of-the-art SMT solvers, only one model is
returned at a time, so for obtaining a complete set of models, the
user must either block an obtained model and re-run the solver to get
the next one with high overhead, or use expert-level
techniques~\cite{bjorner2022user}.
On the other hand, SMT solvers usually come with a rich set of
theories, whereas ASP modulo theories is still limited to basic
theories like difference logic
\cite{DBLP:journals/tplp/JanhunenKOSWS17,DBLP:journals/algorithms/RajaratnamSWCLS23}
and acyclicity constraints \cite{DBLP:conf/lpnmr/BomansonGJKS15}.

\section{Conclusion}
\label{sec:conclusion}

We presented the first approach for synthesising property-rich
inductive predicates for data structures in Separation Logic (SL) from
concrete heap graph examples, by positive-only learning via Answer Set
Programming, with SL-based pruning.
%
%
Our framework \tool is capable of automatically learning predicates
for complex structures with payload constraints and mutual recursion,
facilitating applications of SL-based tools for deductive verification
and program synthesis. 
In the future, we are planning to explore other possible applications
of our predicate synthesiser for program
repair~\cite{Tonder-LeGoues:ICSE18}, program
comprehension~\cite{DBLP:conf/iwpc/BoockmannL22}, and Computer Science
education~\cite{marron2012abstracting}.

\begin{acks}
  We thank Vladimir Gladshtein, Yunjeong Lee, Peter
  Müller, Hila Peleg, George Pîrlea, and Qiyuan Zhao for their feedback on
  earlier drafts of this paper.
  We also thank the reviewers of OOPSLA'25 for their constructive and
  insightful comments.
  This work was partially supported by a Singapore Ministry of
  Education (MoE) Tier 3 grant ``Automated Program Repair''
  MOE-MOET32021-0001.
\end{acks}

\section*{Data Availability}

The implementations of \tool, \ggen, and the benchmark harness
necessary for reproducing our experimental results in
\autoref{sec:evaluation} are available online~\cite{sippy-artefact}. 





\bibliography{references}


\appendix
\label{sec:appendix}

\section{Answer Set Programming in \tool and Graph \ggen.}
In this appendix, we provide a brief overview of the ASP encoding in \tool and \ggen. The notations are standard in ASP, and we will explain them in our context without describing the semantics. An interested reader can refer to the ASP textbook~\cite{lifschitz2019answer} or a tutorial~\cite{aspguide}.

\subsection{Minimisation of Redundant Predicates}
\label{app:minimisation}

The minimisation procedure in \autoref{sec:normalise} of the main text is easily encoded by
the following constraint:
\begin{minted}[fontsize=\small]{prolog}
  :- entail(H, H0), entail(H0, H), lit(H0)<lit(H), output(H).
\end{minted}
The ASP definition of the \pcode{entail(H, H0)} is provided as part of
our framework for several common first-order theories (\eg, integers,
finite sets, and lists).

\subsection{ASP-based Pruning by SL Semantics}
\label{app:slsemantics}

The listed constraints present SL semantics based pruning shown in \autoref{sec:semantics} of the main paper.
\begin{enumerate}
\item \emph{Basic Reachability}: for a pointer domain $P$ and a node $A$, if another node is reached from $A$ by a pointer in $P$, then $A$ must be the "this" node.
\begin{minted}[fontsize=\small]{prolog}
  :- pointer(P), body_lit(_, P, (A, _)), not this(A).
\end{minted}
\item \emph{Basic Assumptions}: "this" node in the basic case is either the null node or equal to another node.
\begin{minted}[fontsize=\small]{prolog}
:- this(X), not null_base(X), not eq_base(X, _).
\end{minted}
\item \emph{Restricted use of} \code{null}: if a node $A$ is a null pointer in clause $T$, then it can only appear twice in the literals of the clause.
\begin{minted}[fontsize=\small]{prolog}
:- body_lit(T, nullptr, (A,)), #count{P,Vars : var_in_literal(T,P,Vars,A)} != 2.
\end{minted}
\item \emph{Quasi-well-founded recursion of payload}: for the pure variable $B1$ in the head of the clause $T$, it should not be less than the pure variable $B2$ in the body of the clause $T$ in partial order.
\begin{minted}[fontsize=\small]{prolog}
:- pure_in_head(T, B1), pure_in_body(T, B2), not partial_le(T, 1, B2, B1).
\end{minted}
Note that by default in \tool, the variables of set and list theories follow the partial order, but not for integer theory: if the integer relation is to describe height or length, the \pcode{po_type(int)} predicate can be manually added as input to accelerate the search (458s to 130s for the balanced tree predicate). Actually we have different parameters in \tool for more practical uses.
\item \emph{Heap functionality}: there should be at most one pointer reached from a single domain $P$ in a clause $T$.
\begin{minted}[fontsize=\small]{prolog}
:- clause(T), pointer(P),
   #count{Var : var_in_body_pos(T, P, _, Var)} > 1.
\end{minted}
\end{enumerate}

\subsection{Auxiliary Placeholders}
\label{app:auxiliary}

On the implementation level for the auxiliary placeholder introduced in \autoref{sec:auxiliary} of the main text, this requires adding an ASP constraint
that forces the parameter of the placeholder predicate (\pcode{Y}
here) to appear \emph{twice} in the whole clause, so it could be later
translated into a single occurrence of a free variable as follows:

\begin{minted}[fontsize=\small]{prolog}
  :- body_lit(T, anynumber, (A,)),
     #count{P,Vars : var_in_literal(T,P,Vars,A)} != 2.
\end{minted}

\vspace{1em}
To sum up, the declarative encoding of ASP gives us a compositional way to encode the different domain knowledge, then achieves the efficient search by the advanced ASP solving.

\end{document}